\documentclass[12pt,twoside]{article}
\usepackage{jgg1}
\usepackage{graphicx,microtype,color}      
\usepackage{amssymb,latexsym,amsmath} 
\usepackage{psfrag,overpic,contour,amsthm}
\usepackage{url}

\newtheorem{thm}{Theorem}
\newtheorem{cor}[thm]{Corollary}
\newtheorem{lem}[thm]{Lemma}
\theoremstyle{definition}
\newtheorem{dft}{Definition}
\theoremstyle{remark}
\newtheorem{rem}{Remark}

\def\Square{\vbox{\hrule\hbox{\vrule height 1.8mm\hskip 1.8mm%
                         \vrule height 1.8mm}\hrule}}

\begin{document}
\begin{JGGarticle} 
        {Isometric Billiards in Ellipses and\\ Focal Billiards in Ellipsoids}
        {H.\ Stachel: Isometric Billiards in Ellipses and Focal Billiards in Ellipsoids}
        {Hellmuth Stachel}
        {\JGGaddress{Institute of Discrete Mathematics and Geometry\\
           Vienna University of Technology\\
           Wiedner Hauptstr.\ 8-10/113, A-1040 Wien, Austria\\
           email: {\tt stachel\at dmg.tuwien.ac.at}}
}           

\newcommand\pfad{} 
\newcommand\Lemref[1]{Lemma~\ref{#1}}
\newcommand\Thmref[1]{Theorem~\ref{#1}}
\newcommand\Figref[1]{Figure~\ref{#1}}
\newcommand\Corref[1]{Corollary~\ref{#1}}
\newcommand\Defref[1]{Definition~\ref{#1}}
\def\ol#1{\overline{#1}}
\def\wt#1{\widetilde{#1}}
\def\RR{{\mathbb R}}
\def\ZZ{{\mathbb Z}}
\def\NN{{\mathbb N}}
\newcommand\Vkt[1]{{\mathbf #1}}
\def\wkl{\,<\mskip-10mu)\mskip4mu}
\def\Frac#1#2{{\displaystyle\frac{#1}{#2}}}
\def\smFrac#1#2{\mbox{\small $\displaystyle\frac{#1}{#2}$ }}
\def\ssmFrac#1#2{\mbox{\footnotesize $\displaystyle \frac{#1}{#2}$}}
\def\Frac#1#2{{\displaystyle\frac{#1}{#2}}}
\def\smFrac#1#2{\mbox{\small $\displaystyle\frac{#1}{#2}$}}
\def\ssmFrac#1#2{\mbox{\footnotesize $\displaystyle \frac{#1}{#2}$}}
\def\Sum{\displaystyle\sum}
\newcommand\name[1]{{\sc #1}}
\newcommand\Ecal{\mathcal E}
\newcommand\Hcal{\mathcal H}
\def\strqu{^{\prime\mskip 2mu 2}}
\def\zwstrqu{^{\prime\prime\mskip 2mu 2}}
\def\iso{^\ast}
\def\zwi{\mskip 9mu}
\newcommand\SN{\mskip 2mu\mathrm{sn}\mskip 2mu}
\newcommand\Sn{\mskip 2mu\mathrm{sn}}
\newcommand\CN{\mskip 2mu\mathrm{cn}\mskip 2mu}
\newcommand\Cn{\mskip 2mu\mathrm{cn}}
\newcommand\DN{\mskip 2mu\mathrm{dn}\mskip 2mu}
\newcommand\Dn{\mskip 2mu\mathrm{dn}}
\newcommand\const{\mathrm{const.}}
\definecolor{blau}{cmyk}{1.00,0.30,0.00,0.20} 
\def\blue{\color{blau}}
\definecolor{rotcmyk}{cmyk}{0.00,1.00,1.00,0.00}
\def\red{\color{rotcmyk}}
\definecolor{hgrau}{cmyk}{0.00,0.00,0.00,0.20}
\def\hgrau{\color{hgrau}}
\definecolor{dgrau}{cmyk}{0.00,0.00,0.00,0.85}
\def\hgrau{\color{dgrau}}
\definecolor{gruen}{cmyk}{1.00,0.00,0.90,0.00} 
\def\green{\color{gruen}}
\definecolor{gelb}{cmyk}{0.00,0.00,1.00,0.00} 
\def\gelb{\color{gelb}}
\definecolor{gelbb}{cmyk}{0.00,0.00,0.30,0.00} 
\def\gelbb{\color{gelbb}}
\def\white{\color{white}}
\begin{JGGabstract}
Billiards in ellipses have a confocal ellipse or hyperbola as caustic.
The goal of this paper is to prove that for each billiard of one type there exists an isometric counterpart of the other type.
Isometry means here that the lengths of corresponding sides are equal.
The transition between these two isometric billiard can be carried out continuosly via isometric focal billiards in a fixed ellipsoid.
The extended sides of these particular billiards in an ellipsoid are focal axes, i.e., generators of confocal hyperboloids. 
\\
This transition enables to transfer properties of planar billiards to focal billiards, in particular billiard motions and canonical parametrizations.
A periodic planar billiard and its associated Poncelet grid give rise to periodic focal billiards and spatial Poncelet grids. 
If the sides of a focal billiard are materialized as thin rods with spherical joints at the vertices and other crossing points between different sides, then we obtain Henrici's hyperboloid, which is flexible between the two planar limits. 
\\[1mm]
\\[1mm]
\end{JGGabstract}
%
\section{Introduction}
A billiard is the trajectory of a mass point in a domain with ideal physical reflections in the boundary.
Already for two centuries, billiards in ellipses have attracted the attention of mathematicians, beginning with J.-V.\ Poncelet and A.\ Cayley.
In 2005 S.\ Tabachnikov published a book on billiards as integrable systems \cite{Tabach}.
In several publications, V.\ Dragovi\'c and M.\ Radnovi\'c studied billiards in higher dimensions from the viewpoint of dynamical systems. 
Their survey \cite{DR_russ} as well as the introduction in \cite{Tabach} provide insights into the rich history of this topic. 

It is wellknown that the sides of a billiard in an ellipse $e$ are tangent to a confocal conic called {\em caustic}.
If this is an ellipse, we speak of an {\em elliptic} billiard.
Otherwise the billiard is called {\em hyperbolic}. 
Computer animations of billiards in ellipses, which were carried out recently by D.\ Reznik, stimulated a new vivid interest on this well studied topic, where algebraic and analytic methods are meeting.
Periodic billiards in an ellipse $e$ with a fixed caustic $c$ provide the standard example of a Poncelet porism, because the periodicity of a billiard in $e$ with caustic $c$ is independent of the choice of the initial vertex. 
Hence, the continuous variation of the initial vertex defines a socalled {\em billiard motion}, and this was the starting point for Reznik's investigations.
He published a list of more than eighty invariants of periodic billiards in ellipses under billiard motions \cite{80} and, with the coauthors R.\ Garcia, J.\ Koiller and M.\ Helman, also proofs for dozens of invariants.
However, also other authors contributed proofs and detected more invariants, among them A.\ Akopyan, M.\ Bialy, A.\ Chavez-Caliz, R.\ Schwartz, S.\ Tabachnikov (see, e.g., \cite{Ako-Tab, Bialy-Tab, Chavez}).
A long list of further references can be found in \cite{80}.

In the literature on billiards in ellipses, the case of elliptic billiards is more intensively studied than that of hyperbolic billiards.
The goal of this paper is to prove an isometry between elliptic and hyperbolic billiards.
It is shown that there is even a continuous transition between these two types via particular billiards in ellipsoids, called {\em focal billiards}.
This transition preserves the lengths of the billiards' sides and is related to Henrici's flexible hyperboloid, a wellknown example of an overconstrained mechanism.

It needs to be mentioned that $n$-dimensional versions of billiards in quadrics for $n\ge 3$ ellipsoid were already studied by V.\ Dragovi\'c and M.\ Radnovi\'c  within the framework of dynamical systems (see, e.g., \cite{DR_2006,DR_russ}). 
In \cite[Example~2-17]{DR_russ}, focal billiards in ellipsoids are explicitely mentioned, but without any further details.

The underlying paper is organized as follows.
The coming section recalls properties of billiards in ellipses and their associated Poncelet grids.
After a brief comparison of elliptic and hyperbolic versions, Section~3 addresses the isometry between these two types.
The proof is postponed to Section~4 which shows the continuous transition from one type to the other via isometric focal billiards in an ellipsoid.
Finally, in Section~5, the transition is used to transfer results concerning billiard motions and invariants from the plane to three dimensions.
Thus we obtain, in terms of Jacobian elliptic functions, a mapping that sends a square grid together with diagonals to the Poncelet grid of a focal billiard on a one-sheeted hyperboloid.

\section{The geometry of billiards in ellipses}
A family of {\em confocal} central conics is given by
\begin{equation}\label{eq:confocal}
  \frac{x^2}{a^2+k} + \frac{y^2}{b^2+k} = 1, \ \mbox{where} \
    k \in \RR \setminus \{-a^2, -b^2\}
\end{equation}
serves as a parameter in the family.
All these conics share the focal points $F_{1,2} = (\pm d,0)$ with $d^2:= a^2-b^2$.

The confocal family sends through each point $P$ outside the common axes of symmetry two orthogonally intersecting conics, one ellipse and one hyperbola \cite[p.~38]{Conics}.
The parameters $(k_e, k_h)$ of these two conics define the {\em elliptic coordinates} of $P$ with
\[  -a^2 < k_h < -b^2 < k_e\,.
\]
If $(x,y)$ are the cartesian coordinates of $P$, then $(k_e,k_h)$ are the roots of the quadratic equation
\begin{equation}\label{eq:cart_in_ell}
  k^2 + (a^2 + b^2 - x^2 - y^2)k + (a^2 b^2 - b^2 x^2 - a^2 y^2) = 0,
\end{equation}
while conversely
\begin{equation}\label{eq:ell_in_cart}
   x^2 = \frac{(a^2 + k_e)(a^2 + k_h)}{d^2}\,, \quad
    y^2 = -\frac{(b^2 + k_e)(b^2 + k_h)}{d^2}\,.
\end{equation}

Let $(a,b) = (a_c,b_c)$ be the semiaxes of the ellipse $c$ with $k = 0$.
Then, for points $P$ on a confocal ellipse $e$ with semiaxes $(a_e,b_e)$ and $k = k_e > 0$, i.e., exterior to $c$, the standard parametrization yields 
\begin{equation}\label{eq:P_coord}
 \begin{array}{c}
   P = (a_e\cos t,\,b_e\sin t), \ 0 \le t < 2\pi, \ \mbox{with} \ 
   a_e^2 = a_c^2 + k_e, \ b_e^2 = b_c^2 + k_e\,. 
 \end{array}
\end{equation}
For the elliptic coordinates $(k_e,k_h)$ of $P$ follows from \eqref{eq:cart_in_ell} that 
\[ k_e + k_h = a_e^2\cos^2 t + b_e^2\sin^2 t - a_c^2 - b_c^2.
\]
After introducing tangent vectors of $e$ and $c$, namely
\def\arraycolsep{0.6mm}
\begin{equation}\label{eq:te_und_tc}
  \Vkt t_e(t) := (-a_e\sin t,\, b_e\cos t) \zwi\mbox{and}\zwi
  \Vkt t_c(t) := (-a_c\sin t,\, b_c\cos t),
\end{equation}
where $\Vert \Vkt t_e\Vert^2 = \Vert \Vkt t_c\Vert^2 + k_e\,$, we obtain
\begin{equation}\label{eq:k_h}
   k_h = k_h(t) = -(a_c^2\sin^2 t + b_c^2\cos^2 t) = -\Vert\Vkt t_c(t)\Vert^2   
   = -\Vert\Vkt t_e(t)\Vert^2 + k_e
\end{equation}
and
\begin{equation}\label{eq:k_e minus k_h}
  \Vert\Vkt t_e(t)\Vert^2 = k_e - k_h(t)\,.
\end{equation}  
Note that points on the confocal ellipses $e$ and $c$ with the same parameter $t$ have the same coordinate $k_h$.
Consequently, they belong to the same confocal hyperbola (\Figref{fig:Poncelet_grid_n8}).
Conversely, points of $e$ or $c$ on this hyperbola have a parameter out of $\{t, -t, \pi+t, \pi-t\}$.

\medskip
Let $\dots P_1P_2P_3\dots$ be a billiard in the ellipse $e$. 
Then the extended sides intersect at points 
\begin{equation}\label{eq:S_i^j}
  S_i^{(j)}:= \left\{ \begin{array}{rl}
  [P_{i-k-1},P_{i-k}]\cap[P_{i+k},P_{i+k+1}] &\zwi\mbox{for} \ j = 2k,
   \\[0.6mm] 
  [P_{i-k},P_{i-k+1}]\cap[P_{i+k},P_{i+k+1}] &\zwi\mbox{for} \ j = 2k-1,
 \end{array} \right.  
\end{equation}
where $i = \dots,1,2,3,\dots$ and $j= 1,2,\dots$.
These points are distributed on particular confocal conics:
For fixed $j$, there are conics $e^{(j)}$ passing through the points $S_i^{(j)}$.
On the other hand, the points $S_i^{(2)}$, $S_i^{(4)},\dots$ are located on the confocal hyperbola through $P_i$, while $S_i^{(1)}$, $S_i^{(3)}, \dots$ belong to a confocal conic other than $c$ through the contact point between the side line $[P_i,P_{i+1}]$ and the caustic $c$.
This configuration is called the associated {\em Poncelet grid}.

For an $N$-periodic billiard, the set of points $S_i^{(j)}$ is finite.
There are confocal conics $e^{(j)}$ for $j = 1,2,\dots,p$ with $p = \left[\frac{N-3}2\right]$.
If the billiard is centrally symmetric, then the points $S_i^{(j)}$ for $j = \frac{N-2}2$ are at infinity.
Below we summarize some properties of the two types of billiards.
For further details see, e.g., \cite{Sta_I}.

\subsection{Elliptical billiards}

 \begin{figure}[htb] 
  \centering
  \def\sz{\small} 
  \psfrag{P1}[lc]{\contourlength{1.2pt}\contour{white}{\sz\red $P_1$}}
  \psfrag{P2}[lb]{\contourlength{1.2pt}\contour{white}{\sz\red $P_2$}}
  \psfrag{P3}[rb]{\contourlength{1.2pt}\contour{white}{\sz\red $P_3$}}
  \psfrag{P4}[rc]{\contourlength{1.2pt}\contour{white}{\sz\red $P_4$}}
  \psfrag{P5}[rt]{\contourlength{1.2pt}\contour{white}{\sz\red $P_5$}}
  \psfrag{P6}[rt]{\contourlength{1.2pt}\contour{white}{\sz\red $P_6$}}
  \psfrag{P7}[ct]{\contourlength{1.2pt}\contour{white}{\sz\red $P_7$}}
  \psfrag{P8}[lt]{\contourlength{1.2pt}\contour{white}{\sz\red $P_8$}}
  \psfrag{P1'}[lb]{\contournumber{32}\contourlength{1.4pt}\contour{white}{\sz\red $P_1'$}}
  \psfrag{P2'}[lb]{\sz\red $P_2'$}
  \psfrag{P3'}[cb]{\contourlength{1.2pt}\contour{white}{\sz\red $P_3'$}}
  \psfrag{P4'}[rc]{} 
  \psfrag{P5'}[rt]{} 
  \psfrag{P6'}[lt]{} 
  \psfrag{P7'}[lt]{} 
  \psfrag{P8'}[lc]{\contourlength{1.2pt}\contour{white}{\sz\red $P_8'$}}
  \psfrag{Q1}[rt]{\contourlength{1.2pt}\contour{gelbb}{\sz\blue $Q_1$}}
  \psfrag{Q2}[ct]{\contourlength{1.4pt}\contour{gelbb}{\sz\blue $Q_2$}}
  \psfrag{Q3}[lt]{\contourlength{1.2pt}\contour{gelbb}{\sz\blue $Q_3$}}
  \psfrag{Q4}[lb]{\contourlength{1.2pt}\contour{gelbb}{\sz\blue $Q_4$}}
  \psfrag{Q5}[lb]{\contourlength{1.2pt}\contour{gelbb}{\sz\blue $Q_5$}}
  \psfrag{Q6}[lb]{\contourlength{1.2pt}\contour{gelbb}{\sz\blue $Q_6$}}
  \psfrag{Q7}[rb]{\contourlength{1.4pt}\contour{gelbb}{\sz\blue $Q_7$}}
  \psfrag{Q8}[rc]{\contourlength{1.2pt}\contour{gelbb}{\sz\blue $Q_8$}}
  \psfrag{S11}[lb]{\contourlength{1.2pt}\contour{white}{\sz\green $S_1^{(1)}$}}
  \psfrag{S21}[lb]{\contourlength{1.2pt}\contour{white}{\sz\green $S_2^{(1)}$}}
  \psfrag{S31}[rb]{\contourlength{1.2pt}\contour{white}{\sz\green $S_3^{(1)}$}}
  \psfrag{S41}[rt]{\contourlength{1.2pt}\contour{white}{\sz\green $S_4^{(1)}$}}
  \psfrag{S51}[rt]{\contourlength{1.2pt}\contour{white}{\sz\green $S_5^{(1)}$}}
  \psfrag{S61}[lt]{\contourlength{1.2pt}\contour{white}{\sz\green $S_6^{(1)}$}}
  \psfrag{S71}[lt]{\contourlength{1.2pt}\contour{white}{\sz\green $S_7^{(1)}$}}
  \psfrag{S81}[lc]{\contourlength{1.4pt}\contour{white}{\sz\green $S_8^{(1)}$}}
  \psfrag{S12}[lc]{\contourlength{1.2pt}\contour{white}{\sz\green $S_1^{(2)}$}}
  \psfrag{S82}[lb]{\contourlength{1.2pt}\contour{white}{\sz\green $S_8^{(2)}$}}
  \psfrag{c}[rt]{\blue $\boldsymbol{c}$}
  \psfrag{e}[lb]{\red $\boldsymbol{e}$}
  \psfrag{e1}[lb]{\contourlength{1.2pt}\contour{white}{\green $\boldsymbol{e}^{(1)}$}}
  \psfrag{e2}[rt]{\green $\boldsymbol{e}^{(2)}$}
  \includegraphics[width=130mm]{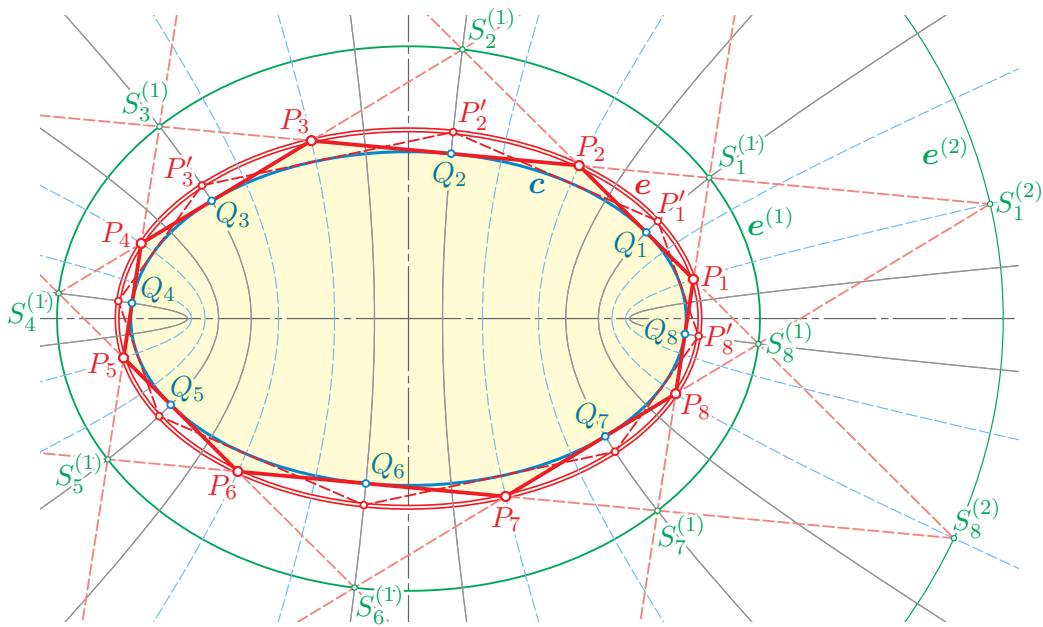} 
  \caption{Periodic billiard $P_1P_2\dots P_8$ with turning number $\tau = 1$ together with the conjugate billiard $P_1'P_2'\dots P_8'$ and the ellipses $e^{(1)}$ and $e^{(2)}$.
The associated billiard in $e^{(1)}$ splits into two conjugate quadrangles; the billiard in $e^{(2)}$ has $\tau = 3$.}
  \label{fig:Poncelet_grid_n8}
\end{figure}

Let $\dots P_1 P_2 P_3\dots$ be a billiard in $e$ with the ellipse $c$ as caustic (\Figref{fig:Poncelet_grid_n8}).
The points $S_i^{(1)}, S_i^{(3)}, \dots$ are located on the confocal hyperbola through the contact point $Q_i$ between $c$ and the side $P_iP_{i+1}$.
On the other hand, $S_i^{(2)}, S_i^{(4)}, \dots$ are located on the confocal hyperbola through $P_i$.
For fixed $j$, the points $S_i^{(j)}$, $i = \dots,1,2,\dots$, are either at infinity or points of a confocal ellipse $e^{(j)}$.
The ellipses $e^{(j)}$ are independent of the position of the initial vertex $P_1\in e$, hence motion invariant. 

Based on the standard parametrizations of $e$ and $c$, let $t_i$ denote the parameter of $P_i$ and $t_i'$ that of $Q_i$.
Then for $N$-periodic elliptic billiards the sequence of parameters $t_1, t_1', t_2, t_2',\dots, t_N'$ is cyclic.
Exchanging $t_i$ with $t_i'$ is equivalent to the exchange of the billiard $P_1 P_2 \dots P_N$ with the {\em conjugate} billiard $P_i' P_2' \dots P_N'$ (see \cite[Definition~3.10]{Sta_I}).
The turning number $\tau$ of any periodic elliptic billiard counts the surroundings of $e$ within one period.
If the original billiard has $\tau=1$, then the extended sides in $e^{(1)}$ determine a billiard with turning number $\tau=2$, in $e^{(2)}$ with $\tau=3$, and so on.

For even $N$, opposite vertices of the elliptic billiard belong to the same confocal hyperbola.
For odd $N$, each vertex and the contact point of $c$ with the opposite side belong to the same confocal hyperbola.
As a consequence, for even $N$ and odd $\tau$, the elliptic billiard is centrally symmetric (Figures~\ref{fig:Poncelet_grid_n8} or \ref{fig:isometric2}, bottom).
 For odd $N$ and odd $\tau$, the billiard is centrally symmetric to the conjugate billiard. 
 If $N$ is odd and $\tau$ even, then the conjugate billiard coincides with the original one.
 
\subsection{Hyperbolic billiards}

As illustrated in \Figref{fig:hyp_Kaust}, billiards in an ellipse $e$ with a confocal hyperbola $c$ as caustic are zig-zags between an upper and lower subarc of $e$.
However, they differ from elliptic billiard also in other respects.

\begin{figure}[t] 
  \centering 
  \def\sz{\small}
  \def\ssz{\footnotesize}
  \psfrag{Q1}[lc]{\contourlength{1.5pt}\contour{white}{\ssz\blue $T_1$}}
  \psfrag{Q4}[rc]{\contourlength{1.2pt}\contour{white}{\ssz\blue $T_4$}}
  \psfrag{Q5}[rt]{\contourlength{1.2pt}\contour{white}{\ssz\blue $T_5$}}
  \psfrag{Q6}[rb]{\contourlength{1.2pt}\contour{white}{\ssz\blue $T_6$}}
  \psfrag{Q7}[rc]{\contourlength{1.5pt}\contour{white}{\ssz\blue $T_7$}}
  \psfrag{Q10}[lc]{\contourlength{1.2pt}\contour{white}{\ssz\blue $T_{10}$}}
  \psfrag{Q11}[lt]{\contourlength{1.2pt}\contour{white}{\ssz\blue $T_{11}$}}
  \psfrag{Q12}[lb]{\contourlength{1.2pt}\contour{white}{\ssz\blue $T_{12}$}}
  \psfrag{P1}[lb]{\contourlength{1.2pt}\contour{white}{\sz\red $P_1$}}
  \psfrag{P2}[lt]{\contourlength{1.2pt}\contour{white}{\sz\red $P_2$}}
  \psfrag{P3}[lb]{\contourlength{1.2pt}\contour{white}{\sz\red $P_3$}}
  \psfrag{P4}[lt]{\contourlength{1.2pt}\contour{white}{\sz\red $P_4$}}
  \psfrag{P5}[lb]{\contourlength{1.2pt}\contour{white}{\sz\red $P_5$}}
  \psfrag{P6}[rt]{\contourlength{1.2pt}\contour{white}{\sz\red $P_6$}}
  \psfrag{P7}[rb]{\contourlength{1.2pt}\contour{white}{\sz\red $P_7$}}
  \psfrag{P8}[rt]{\contourlength{1.2pt}\contour{white}{\sz\red $P_8$}}
  \psfrag{P9}[rb]{\contourlength{1.2pt}\contour{white}{\sz\red $P_9$}}
  \psfrag{P10}[rt]{\contourlength{1.2pt}\contour{white}{\sz\red $P_{10}$}}
  \psfrag{P11}[rb]{\contourlength{1.2pt}\contour{white}{\sz\red $P_{11}$}}
  \psfrag{P12}[lt]{\contourlength{1.2pt}\contour{white}{\sz\red $P_{12}$}}
  \psfrag{S1^1}[rt]{\contourlength{1.2pt}\contour{white}{\ssz\green $S_2^{(1)}$}}
  \psfrag{S2^1}[rt]{\contourlength{1.2pt}\contour{white}{\ssz\green $S_3^{(1)}$}}
  \psfrag{S3^1}[rt]{\contourlength{1.2pt}\contour{white}{\ssz\green $S_4^{(1)}$}}
  \psfrag{S4^1}[lb]{\contourlength{1.2pt}\contour{white}{\ssz\green $S_5^{(1)}$}}
  \psfrag{S5^1}[lc]{\contourlength{1.2pt}\contour{white}{\ssz\green $S_6^{(1)}$}}
  \psfrag{S6^1}[lb]{\contourlength{1.2pt}\contour{white}{\ssz\green $S_7^{(1)}$}}
  \psfrag{S7^1}[rt]{\contourlength{1.2pt}\contour{white}{\ssz\green $S_8^{(1)}$}}
  \psfrag{S8^1}[rt]{\contourlength{1.2pt}\contour{white}{\ssz\green $S_9^{(1)}$}}
  \psfrag{S9^1}[rt]{\contourlength{1.2pt}\contour{white}{\ssz\green $S_{10}^{(1)}$}}
  \psfrag{S10^1}[rc]{\contourlength{1.2pt}\contour{white}{\ssz\green $S_{11}^{(1)}$}}
  \psfrag{S11^1}[rt]{\contourlength{1.2pt}\contour{white}{\ssz\green $S_{12}^{(1)}$}}
  \psfrag{S12^1}[lc]{\contourlength{1.2pt}\contour{white}{\ssz\green $S_1^{(1)}$}}
  \psfrag{S1^2}[lb]{\contourlength{1.2pt}\contour{white}{\ssz\green $S_3^{(2)}$}}
  \psfrag{S2^2}[lt]{\contourlength{1.2pt}\contour{white}{\ssz\green $S_4^{(2)}$}}
  \psfrag{S3^2}[lb]{\contourlength{1.2pt}\contour{white}{\ssz\green $S_5^{(2)}$}}
  \psfrag{S4^2}[rc]{\contourlength{1.2pt}\contour{white}{\ssz\green $S_6^{(2)}$}}
  \psfrag{S5^2}[rc]{\contourlength{1.2pt}\contour{white}{\ssz\green $S_7^{(2)}$}}
  \psfrag{S6^2}[rb]{\contourlength{1.2pt}\contour{white}{\ssz\green $S_8^{(2)}$}}
  \psfrag{S7^2}[rb]{\contourlength{1.2pt}\contour{white}{\ssz\green $S_9^{(2)}$}}
  \psfrag{S8^2}[rt]{\contourlength{1.2pt}\contour{white}{\ssz\green $S_{10}^{(2)}$}}
  \psfrag{S9^2}[rc]{\contourlength{1.2pt}\contour{white}{\ssz\green $S_{11}^{(2)}$}}
  \psfrag{S10^2}[lb]{\contourlength{1.2pt}\contour{white}{\ssz\green $S_{12}^{(2)}$}}
  \psfrag{S11^2}[lt]{\contourlength{1.2pt}\contour{white}{\ssz\green $S_1^{(2)}$}}
  \psfrag{S12^2}[rt]{\contourlength{1.2pt}\contour{white}{\ssz\green $S_2^{(2)}$}}
  \psfrag{S1^3}[rt]{\contourlength{1.2pt}\contour{white}{\ssz\green $S_3^{(3)}$}}
  \psfrag{S2^3}[rt]{\contourlength{1.2pt}\contour{white}{\ssz\green $S_4^{(3)}$}}
  \psfrag{S3^3}[ct]{\contourlength{1.2pt}\contour{white}{\ssz\green $S_5^{(3)}$}}
  \psfrag{S4^3}[lc]{\contourlength{1.2pt}\contour{white}{\ssz\green $S_6^{(3)}$}}
  \psfrag{S5^3}[lt]{\contourlength{1.2pt}\contour{white}{\ssz\green $S_7^{(3)}$}}
  \psfrag{S6^3}[rt]{\contourlength{1.2pt}\contour{white}{\ssz\green $S_8^{(3)}$}}
  \psfrag{S7^3}[rt]{\contourlength{1.2pt}\contour{white}{\ssz\green $S_9^{(3)}$}}
  \psfrag{S8^3}[rt]{\contourlength{1.2pt}\contour{white}{\ssz\green $S_{10}^{(3)}$}}
  \psfrag{S9^3}[ct]{\contourlength{1.2pt}\contour{white}{\ssz\green $S_{11}^{(3)}$}}
  \psfrag{S10^3}[rc]{\contourlength{1.2pt}\contour{white}{\ssz\green $S_{12}^{(3)}$}}
  \psfrag{S11^3}[rt]{\contourlength{1.2pt}\contour{white}{\ssz\green $S_1^{(3)}$}}
  \psfrag{S12^3}[lb]{\contourlength{1.2pt}\contour{white}{\ssz\green $S_2^{(3)}$}}
  \psfrag{S1^5}[rt]{\contourlength{1.2pt}\contour{white}{\ssz\green $S_4^{(5)}$}}
  \psfrag{S2^5}[lb]{\contourlength{1.5pt}\contour{white}{\ssz\green $S_5^{(5)}$}}
  \psfrag{S3^5}[rt]{\contourlength{1.2pt}\contour{white}{\ssz\green $S_6^{(5)}$}}
  \psfrag{S4^5}[lb]{\contourlength{1.2pt}\contour{white}{\ssz\green $S_7^{(5)}$}}
  \psfrag{S5^5}[rt]{\contourlength{1.2pt}\contour{white}{\ssz\green $S_8^{(5)}$}}
  \psfrag{S6^5}[rt]{\contourlength{1.2pt}\contour{white}{\ssz\green $S_9^{(5)}$}}
  \psfrag{e}[lb]{\sz\red $\boldsymbol{e}$}
  \psfrag{c}[lb]{\sz\blue $\boldsymbol{c}$}
  \psfrag{e1}[lt]{\contourlength{1.5pt}\contour{white}{\sz $\boldsymbol{e}^{(1)}$}}
  \psfrag{e2}[lb]{\sz $\boldsymbol{e}^{(2)}$} 
  \psfrag{e3}[lt]{\contourlength{1.2pt}\contour{white}{\sz $\boldsymbol{e}^{(3)}$}}
  \psfrag{e5}[rc]{\contourlength{1.2pt}\contour{white}{\sz $\boldsymbol{e}^{(5)}$}}
 \includegraphics[width=150mm]{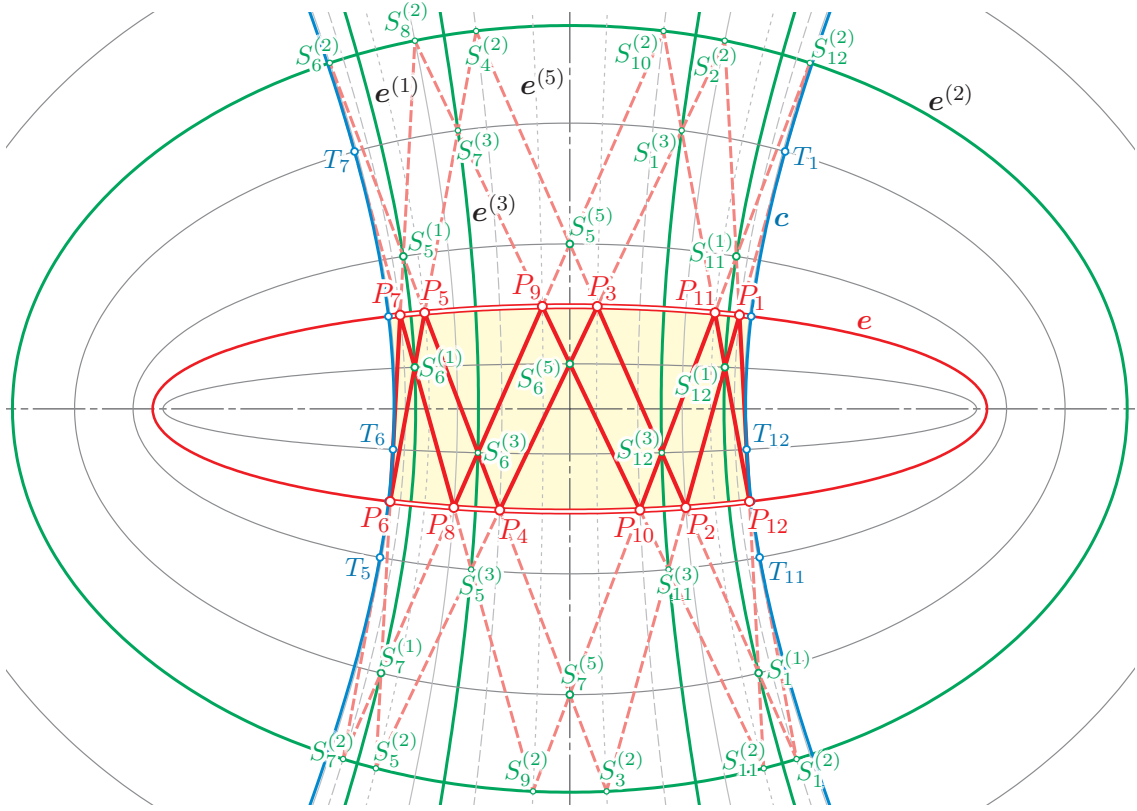}
  \caption{Periodic billiard $P_1P_2\dots P_{12}$ in the ellipse $e$ with the hyperbola $c$ as caustic, together with the hyperbolas $e^{(1)}$, $e^{(3)}$ and the ellipse $e^{(2)}$.
  The associated billiard in $e^{(2)}$ splits into four quadrangles.}
  \label{fig:hyp_Kaust}
\end{figure}

Let $T_i$ be the contact point between $c$ and the side line $[P_i,P_{i+1}]$.
Then the points $S_i^{(1)}, S_i^{(3)}, \dots$ are located on the confocal ellipse through $T_i$.
On the other hand, $S_i^{(2)}, S_i^{(4)}, \dots$ are located on the confocal hyperbola through $P_i$.
For odd $j$, the points $S_i^{(j)}$, $i = \dots,1,2,\dots$, belong to a confocal hyperbola $e^{(j)}$ or to the secondary axis of $e$, or they are at infinity.
If $S_i^{(j)}$ is finite, then the tangent to $e^{(j)}$ bisects the angle between the sides lines through $S_i^{(j)}$.
For even $j$, the $S_i^{(j)}$ are vertices of a billiard inscribed in an ellipse $e^{(j)}$ (note $e^{(2)}$ in \Figref{fig:hyp_Kaust}) or at infinity.
All conics $e^{(j)}$ are motion invariant. 

Now we concentrate on $N$-periodic hyperbolic billiards $P_1 P_2\dots P_N$.
Since they oscillate between the upper and lower section of the circumscribed ellipse, $N$ must be even.
If $\ol P_1 \ol P_2\dots \ol P_N$ denotes the billiard's image under the reflection in the principal axis, then the sequence of parameters $t_1, \ol t_2, t_3, \ol t_4,\dots,\ol t_N$ of the vertices $P_1, \ol P_2, P_3, \ol P_4,\dots, \ol P_N$ is cyclic.
Also for hyperbolic billiards, it is possible to define the turning number $\tau$.
It counts how often the points $P_1, \ol P_2, P_3,\dots, \ol P_N$ run to and fro along one component of $e$.
According to \cite[Definition~3.13]{Sta_I}, there exist conjugate billiards also in this case.
 
For $N\equiv 0\pmod 4$, the hyperbolic billiards are symmetric with respect to (w.r.t.\ in brief) the secondary axis of $e$ and $c$ (\Figref{fig:hyp_Kaust}). 
For $N\equiv 2\pmod 4$ and odd turning number $\tau$, the billiards are centrally symmetric (\Figref{fig:isometric2}, top).
For even $\tau$, each hyperbolic billiard is symmetric w.r.t.\ the principal axis of $e$ and $c$ (note also \Lemref{lem:spat_bil_symm}).  
If the initial point $P_1$ is chosen at any point of intersection between $e$ and the hyperbola $c$, then the billiard is twofold covered.
The first side $P_1P_2$ is of course tangent to $c$ at $P_1$. 

\subsection{Isometry between elliptic and hyperbolic billiards}

\begin{figure}[hbt] 
  \centering
  \def\sz{\small} 
  \psfrag{01}[lb]{\contourlength{1.2pt}\contour{white}{\sz\red $P_1''$}}
  \psfrag{02}[rc]{\contourlength{1.6pt}\contour{white}{\sz\red $P_2''$}}
  \psfrag{03}[rb]{\contourlength{1.2pt}\contour{white}{\sz\red $P_3''$}}
  \psfrag{04}[ct]{\contourlength{1.2pt}\contour{white}{\sz\red $P_4''$}}
  \psfrag{05}[cb]{\contourlength{1.2pt}\contour{white}{\sz\red $P_5''$}}
  \psfrag{06}[rt]{\contourlength{1.2pt}\contour{white}{\sz\red $P_6''$}}
  \psfrag{07}[rc]{\contourlength{1.6pt}\contour{white}{\sz\red $P_7''$}}
  \psfrag{08}[rt]{\contourlength{1.2pt}\contour{white}{\sz\red $P_8''$}}
  \psfrag{09}[lc]{\contourlength{1.6pt}\contour{white}{\sz\red $P_9''$}}
  \psfrag{10}[ct]{\contourlength{1.2pt}\contour{white}{\sz\red $P_{10}''$}}
  \psfrag{11}[rb]{\contourlength{1.2pt}\contour{white}{\sz\red $P_{11}''$}}
  \psfrag{12}[rt]{\contourlength{1.2pt}\contour{white}{\sz\red $P_{12}''$}}
  \psfrag{13}[cb]{\contourlength{1.2pt}\contour{white}{\sz\red $P_{13}''$}}
  \psfrag{14}[lc]{\contourlength{1.6pt}\contour{white}{\sz\red $P_{14}''$}}
  \psfrag{Q1}[lc]{\contourlength{1.5pt}\contour{white}{\sz\blue $Q_1''$}}
  \psfrag{T1}[rb]{\contourlength{1.2pt}\contour{white}{\sz\blue $T_1''$}}
  \psfrag{S32}[lt]{\contourlength{1.2pt}\contour{white}{\sz\blue $S_3^{(2)\mskip 1mu\prime\prime}$}}
  \psfrag{S54}[lc]{\contourlength{1.2pt}\contour{white}{\sz\blue $S_5^{(4)\mskip 1mu\prime\prime}$}}
  \psfrag{S64}[ct]{\contourlength{1.2pt}\contour{white}{\sz\blue $S_6^{(4)\mskip 1mu\prime\prime}$}}
  \psfrag{e"}[lt]{\contourlength{1.2pt}\contour{white}{\sz\red $\boldsymbol{e}''$}}
  \psfrag{c"}[lb]{\contourlength{1.2pt}\contour{white}{\sz\blue $\boldsymbol{c}''$}}
  \includegraphics[width=140mm]{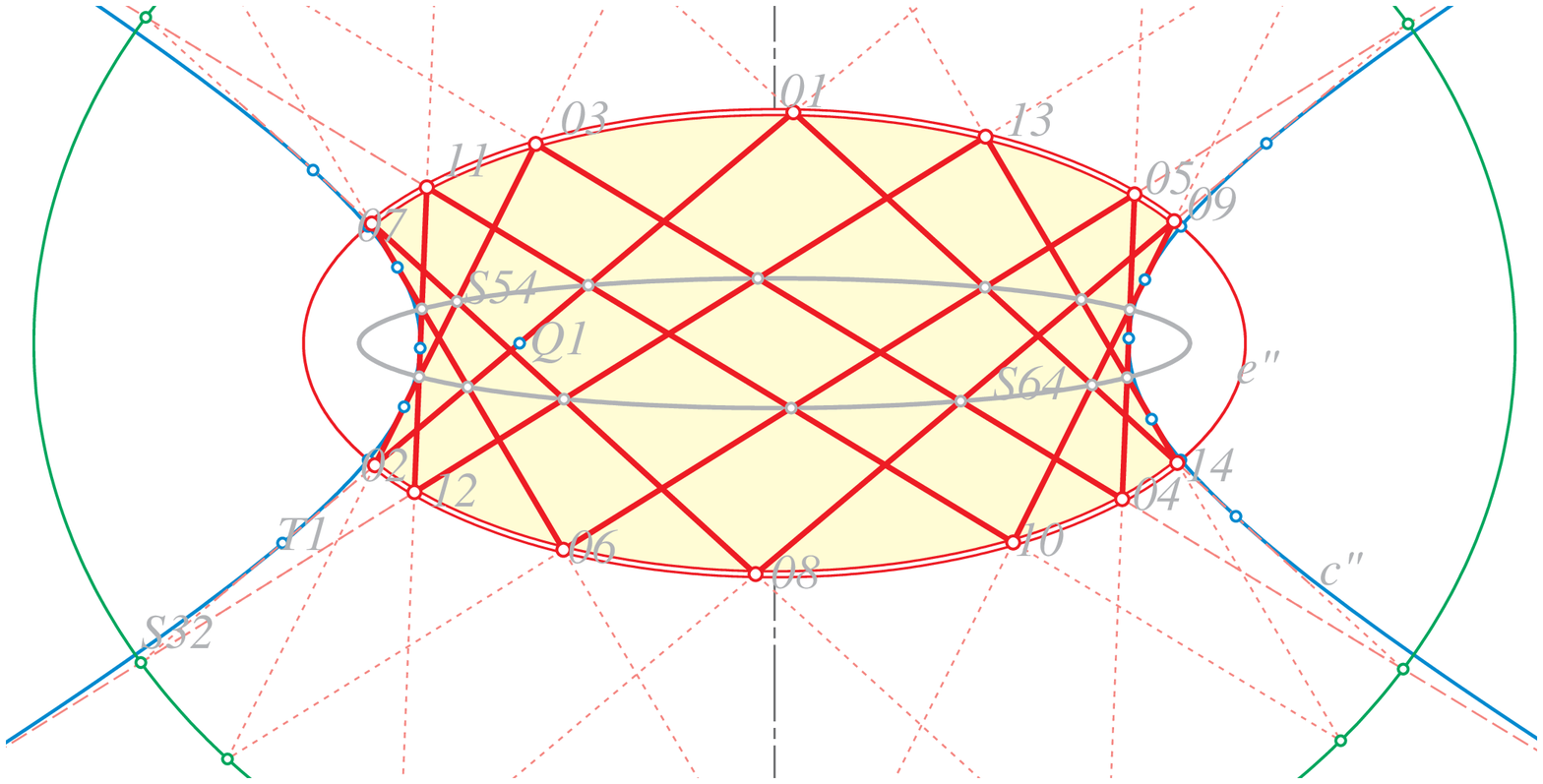} 

\bigskip
  \psfrag{01}[lb]{\contourlength{1.2pt}\contour{white}{\sz\red $P_1'$}}
  \psfrag{02}[rc]{\contourlength{1.2pt}\contour{white}{\sz\red $P_2'$}}
  \psfrag{03}[rt]{\contourlength{1.2pt}\contour{white}{\sz\red $P_3'$}}
  \psfrag{04}[lt]{\contourlength{1.2pt}\contour{white}{\sz\red $P_4'$}}
  \psfrag{05}[lb]{\contourlength{1.2pt}\contour{white}{\sz\red $P_5'$}}
  \psfrag{06}[rb]{\contourlength{1.2pt}\contour{white}{\sz\red $P_6'$}}
  \psfrag{07}[rc]{\contourlength{1.2pt}\contour{white}{\sz\red $P_7'$}}
  \psfrag{08}[ct]{\contourlength{1.2pt}\contour{white}{\sz\red $P_8'$}}
  \psfrag{09}[lc]{\contourlength{1.2pt}\contour{white}{\sz\red $P_9'$}}
  \psfrag{10}[lb]{\contourlength{1.2pt}\contour{white}{\sz\red $P_{10}'$}}
  \psfrag{11}[rb]{\contourlength{1.2pt}\contour{white}{\sz\red $P_{11}'$}}
  \psfrag{12}[rt]{\contourlength{1.2pt}\contour{white}{\sz\red $P_{12}'$}}
  \psfrag{13}[lt]{\contourlength{1.2pt}\contour{white}{\sz\red $P_{13}'$}}
  \psfrag{14}[lc]{\contourlength{1.2pt}\contour{white}{\sz\red $P_{14}'$}}
  \psfrag{Q1}[lt]{\contourlength{1.2pt}\contour{white}{\sz\blue $Q_1'$}}
  \psfrag{T1}[cb]{\contourlength{1.2pt}\contour{white}{\sz\blue $T_1'$}}
  \psfrag{S32}[lt]{\contourlength{1.2pt}\contour{white}{\sz\blue $S_3^{(2)\mskip 1mu\prime}$}}
  \psfrag{S54}[lc]{\contourlength{1.2pt}\contour{white}{\sz\blue $S_5^{(4)\mskip 1mu\prime}$}}
  \psfrag{S64}[cb]{\contourlength{1.2pt}\contour{white}{\sz\blue $S_6^{(4)\mskip 1mu\prime}$}}
  \psfrag{e'}[lt]{\contourlength{1.2pt}\contour{white}{\sz\red $\boldsymbol{e}'$}}
  \psfrag{c'}[rt]{\contourlength{1.2pt}\contour{white}{\sz\blue $\boldsymbol{c}'$}}
  \includegraphics[width=140mm]{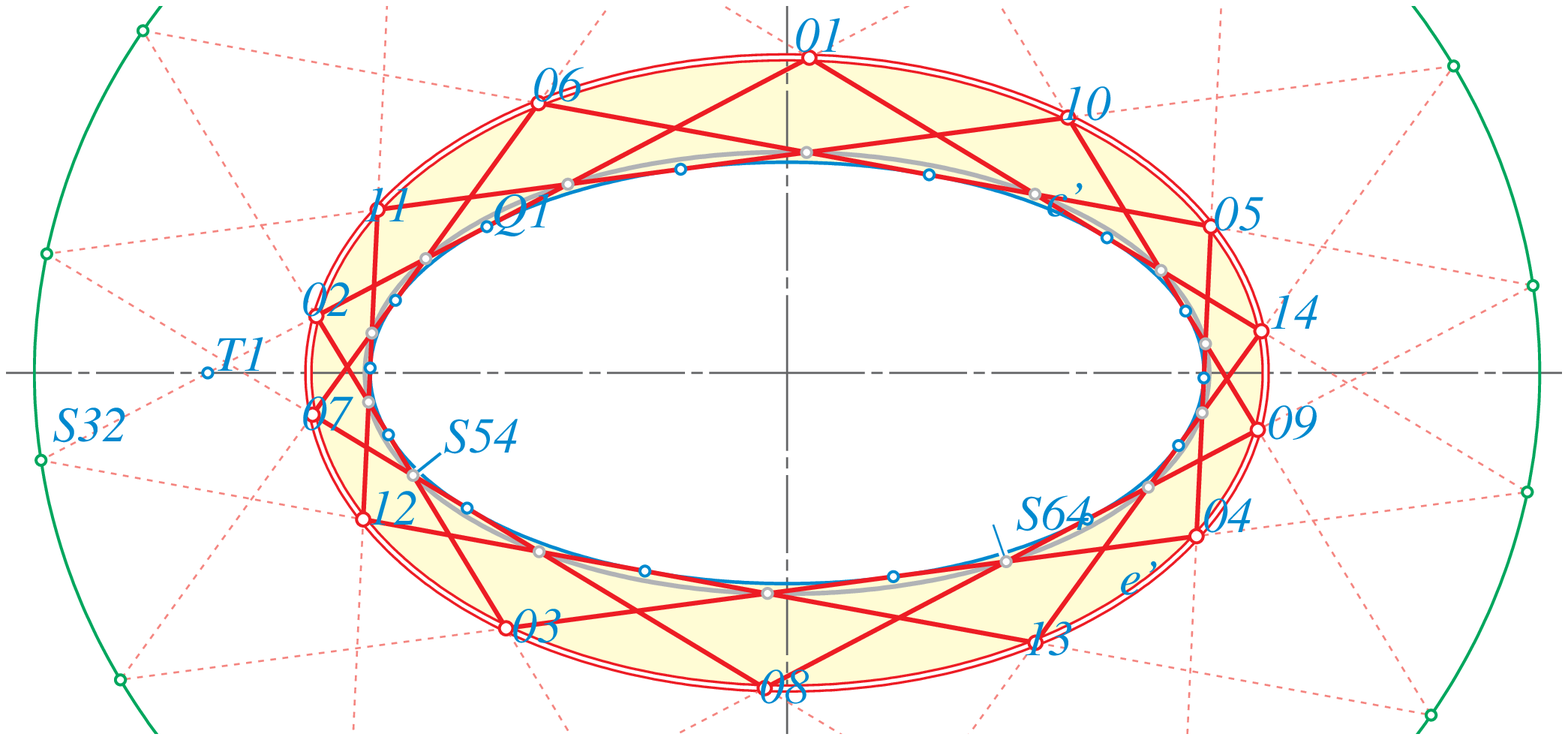} 
  \caption{Two isometric periodic billiards with $N=14$ and turning number $\tau=3$.}  
  \label{fig:isometric2}
\end{figure}

There is a surprising relation between billiards with ellipses and hyperbolas as caustics.

\begin{dft}\label{def:isometric} 
Two polygons $P_1P_2\dots$ and $P_1\iso P_2\iso \dots$ in the Euclidean 3-space are called {\em isometric} if corresponding sides $P_iP_{i+1}$ and $P_i\iso P_{i+1}\iso$ have equal lengths for all $i=1,2,\dots$
\end{dft}

\begin{thm}\label{thm:isometric}
\begin{enumerate}
\item For each billiard $P_1' P_2' \dots$ in an ellipse $e'$ with an ellipse $c'$ as caustic there exists an isometric billiard $P_1'' P_2'' \dots$ in an ellipse $e''$ with a hyperbola $c''$ as caustic.
The billiard inscribed in $e''$ is unique only up to a reflection in the principal axis.
\item Conversely, for each billiard with a hyperbola $c''$ as caustic there exists an isometric billiard with an ellipse $c'$ as caustic, provided that in the case of an $N$-periodic billiard with $N\equiv 2\pmod 4$, we traverse the billiard with the elliptic caustic twice.
\item For each side $P_i' P_{i+1}'$ of the original billiard, the isometry $[P_i',P_{i+1}'] \to [P_i'',P_{i+1}'']$ between the extended sides maps the incident points $S_k^{(j)\mskip 1mu\prime}$ with $j\equiv 0\pmod 2$ of the associated Poncelet grid to the corresponding point $S_k^{(j)\mskip 1mu\prime\prime}$ of the isometric billiard, hence 
\[  \ol{P_{i}'' S_k^{(j)\mskip 1mu\prime\prime}} = \ol{P_{i}' S_k^{(j)\mskip 1mu\prime}} \quad\mbox{for}\quad k = i+\ssmFrac j2\,. 
\]
\item For all $i$, the isometric image of the contact point $Q_i'$ of the side $P_i' P_{i+1}'$ with the ellipse $c'$ is the point of intersection $Q_i''$ of the side $P_i'' P_{i+1}''$ with the principal axis.
On the other hand, the isometry sends the point of intersection $T_i'$ between the side line $[P_i', P_{i+1}']$ and the principal axis to the contact point $T_i''$ of the side $P_i'' P_{i+1}''$ with the hyperbola $c''$ (\Figref{fig:isometric2}).
\end{enumerate}
\end{thm}

We use the notations $(a_e',b_e')$ and $(a_c',b_c')$ for the semiaxes of the ellipses $e'$ and $c'$ as well as $(a_e'',b_e'')$ and $(a_c'', b_c'')$ for those of $e''$ and the confocal hyperbola $c''$.
Finally, let $d'$ and $d''$ denote the respective eccentricities of the confocal pairs $(e',c')$ and $(e'',c'')$ and $k_e'$ as well as $k_e''$ the respective elliptic coordinates of $e'$ and $e''$ w.r.t.\ the caustics, i.e., $k_e' = a_e\strqu - a_c\strqu$ and $k_e'' = a_e\zwstrqu - a_c\zwstrqu$.
Then there hold the symmetric relations
\begin{equation}\label{eq:iso}
 \def\arraycolsep{0.5mm}
 \begin{array}{rclrclrclrclrcl}
  a_c'' &= &d', \quad &b_c'' &= &b_c', \quad &d'' &= &a_c', \quad &a_e'' &= &a_e', 
   \quad &b_e\zwstrqu &= &k_e'\,,
  \\[1mm]
  a_c' &= &d'', &b_c' &= &b_c'', &d' &= &a_c'', &a_e' &= &a_e'', 
   &b_e\strqu &= &k_e''\,.
 \end{array}  
\end{equation} 
Moreover, the distances of corresponding points $P_i'$ and $P_i''$ from the respective secondary axes are proportional.

Item~4 means 
that
\[  \ol{P_i' Q_i'} = \ol{P_i'' Q_i''} \zwi \mbox{and}\zwi 
    \ol{P_i' Q_{i-1}'} = \ol{P_i'' Q_{i-1}''}
\]
with $Q_i''$ and $Q_{i-1}''$ as points of intersection of the principal axis with the sides through $P_i''\,$.

\medskip
Instead of a verification of \Thmref{thm:isometric} based on formulas from \cite{Sta_I}, we embed below the two isometric planar billiards as limiting poses in a continuous set of isometric spatial billiards in an ellipsoid.

\section{Focal billiards in an ellipsoid}

Let the billiard in the ellipse $e'$ with the confocal ellipse $c'$ as caustic be placed in the plane $z=0$.
There exists an ellipsoid $\mathcal E$ through $e'$ with $c'$ as focal ellipse.
Hence, by \cite[p.~280]{Quadrics}, $\mathcal E$ has the semiaxes 
\begin{equation}\label{eq:ae,be,ce}
  a_e = a_e', \quad b_e = b_e', \quad c_e = \sqrt{a_e\strqu - a_c\strqu} = 
  \sqrt{k_e'}\,,\zwi\mbox{where}\zwi a_e > b_e > c_e.
\end{equation}
All quadrics which are confocal with $\mathcal E$ can be represented as
\begin{equation}\label{eq:confocal_Raum}  
  \frac{x^2}{a_c\strqu + k} + \frac{y^2}{b_c\strqu + k} + \frac{z^2}{k} = 1 \quad \mbox{for} \quad
  k\in \mathbb R\setminus\{-a_c\strqu,-b_c\strqu,0\}.
\end{equation}
We obtain the focal ellipse $c'$ as limiting case for $k=0$ and the ellipsoid $\mathcal E$ for 
\[  k = c_e^2 = k_e' = a_e\strqu - a_c\strqu = b_e\strqu - b_c\strqu.
\]
The family of confocal central quadrics contains  
\begin{equation}\label{eq:confoc_interval}
 \def\arraycolsep{1.0mm}
  \mbox{for} \ \left\{ \begin{array}{rcll}
        0\; < &k=k_0 &< \infty      &\mbox{triaxial ellipsoids,}\\[0.5mm]
 -b_c\strqu < &k=k_1 &< \;0         &\mbox{one-sheeted hyperboloids,}\\[0.5mm]
 -a_c\strqu < &k=k_2 &< -b_c\strqu~ &\mbox{two-sheeted hyperboloids}
  \end{array} \right.
\end{equation}
and the focal hyperbola $c''$ in the plane $y=0$ as the limit for $k = -b_c\strqu$ with the semiaxes $a_c'' = \sqrt{a_c\strqu - b_c\strqu} = d'$ and $b_c'' = b_c'\,$.

We recall that the family of confocal quadrics sends through each point $P = (x,y,z)$ with $xyz \ne 0$ three mutually orthogonal surfaces, one of each type.
The parameters $(k_0,k_1,k_2)$ of these quadrics are the {\em elliptic coordinates} of $P$ and satisfy
\begin{equation}\label{eq:elliptic-cartesian}
   x^2 = \frac{(a_c\strqu\!+k_0)(a_c\strqu\!+k_1)(a_c\strqu\!+k_2)}
          {(a_c\strqu - b_c\strqu)a_c\strqu}, \ \,
   y^2 = \frac{(b_c\strqu\!+k_0)(b_c\strqu\!+k_1)(b_c\strqu\!+k_2)}
          {b_c\strqu(b_c\strqu - a_c\strqu)}, \ \,
   z^2 = \frac{k_0 k_1 k_2}{a_c\strqu b_c\strqu}.
\end{equation}
Conversely, eight points in space, symmetrically placed w.r.t. the coordinate
frame, share their elliptic coordinates $(k_0,k_1,k_2)$.

In the $[x,y]$-plane, the traces of the confocal triaxial ellipsoids and the one-sheeted hyperboloids are ellipses confocal with $e'$ and respectively outside or inside the focal ellipse $c'$. 
The two-sheeted hyperboloids intersect the $[x,y]$-plane along confocal hyperbolas.
Their second elliptic coordinate $k_h'$ according to \eqref{eq:k_h} equals $k_2\,$.

At the point $P = (x,y,z)$ with $xyz \ne 0$ and position vector $\Vkt p$, the normal vectors to the three quadrics through $P$, 
\begin{equation}\label{eq:n1n2n3}
  \Vkt n_i = 
   \left( \Frac{x}{a_c\strqu + k_i}, \ \Frac{y}{b_c\strqu + k_i}, \ \Frac{z}{k_i} \right) 
    \zwi \mbox{for} \zwi k_i\ne 0,-b_c\strqu,-a_c\strqu \zwi \mbox{and} \zwi i=0,1,2 \,,
\end{equation}
are mutually orthogonal and yield the dot product 
\[  \langle \Vkt p, \Vkt n_i\rangle = \frac{x^2}{a_c\strqu + k_i} 
     + \frac{y^2}{b_c\strqu + k_i} + \frac{z^2}{k_i} = 1
\]              
and .
Therefore, confocal quadrics form a triply orthogonal system, and any two confocal quadrics of different types intersect each other along a line of curvature w.r.t\ both quadrics.

Let $P$ be a common point of the ellipsoid $\mathcal E$ with $k_0 = k_e'$ and a confocal one-sheeted hyperboloid $\mathcal H_1$ given bei $k = k_1$.
Then the two generators of $\mathcal H_1$ through $P$ are symmetric w.r.t.\ $\mathcal E$, since they are asymptotic lines on $\mathcal H_1$ and the tangent plane to $\mathcal H_1$ at $P$ is orthogonal to $\mathcal E$.
In other words , the reflection in $\mathcal E$ exchanges the two generators of $\mathcal H_1$ through $P$ (\Figref{fig:symm_Erz}).

\begin{figure}[htb] 
  \centering
  \def\sz{\small} 
  \psfrag{E}[cc]{\blue $\mathcal E$}
  \psfrag{H}[cc]{\blue $\mathcal H_1$}
  \psfrag{e1}[rb]{\sz\red $\boldsymbol{e}$}
  \psfrag{P}[cb]{\sz $P$}
  \psfrag{th}[cc]{\sz $\theta/2$}
  \psfrag{v2}[lb]{\contourlength{1.2pt}\contour{white}{\sz $\Vkt n_1$}}
  \psfrag{v3}[lt]{\contourlength{1.2pt}\contour{white}{\sz $\Vkt n_2$}}
  \includegraphics[width=110mm]{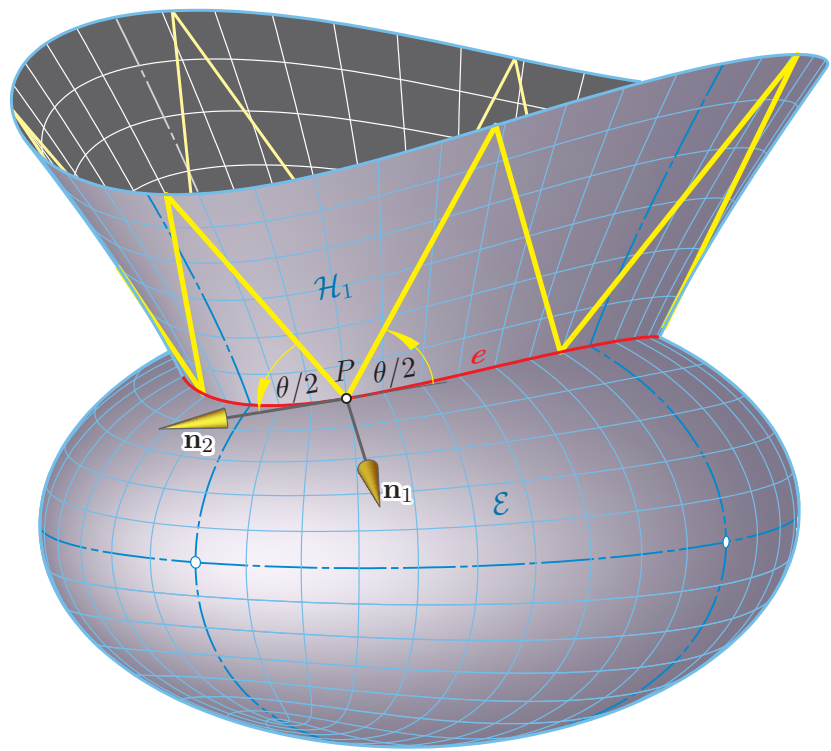} 
  \caption{The reflection in the ellipsoid $\mathcal E$ permutes on each confocal one-sheeted hyperboloid $\mathcal H_1$ the two reguli.
  The yellow polygon is a focal billiard between two confocal ellipsoids.$^1$} 
  \label{fig:symm_Erz} 
\end{figure}

\addtocounter{footnote}{1}
\footnotetext{Billiard between confocal conics and quadrics have already studied in \cite{DR_2006,DR_russ}.}
All generators of the one-sheeted hyperboloids in the confocal family together with the tangent lines of the two focal conics are called {\em focal axes}\footnote{
The notation stems from the analogy to focal points of plane algebraic curves.
They are defined by the property that the isotropic lines through any focal point are tangents of the curve.
Similarly, the isotropic planes through any focal axis are tangent to all quadrics in the confocal family (see \cite[p.~289]{Quadrics}).}
 of the confocal family as well as of any contained single ellipsoid $\mathcal E$.

\begin{dft}\label{def:focal_bill}
A billiard in an ellipsoid $\mathcal E$ is called {\em focal billiard} if each side is located on a focal axis of $\mathcal E$.
\end{dft}

\begin{rem}
There exist also billiards in $\mathcal E$ which are non-focal.
They can even be periodic, for example any periodic billiard inscribed in $e'$, but with a caustic different from the focal ellipse $c'$ of $\mathcal E$.
The side lines of any non-focal billiard in the ellipsoid are tangents to two fixed quadrics, which are confocal with $\mathcal E$ (see \cite[p.~332]{Quadrics}).  
\end{rem}

\begin{lem}\label{lem:focal_bill} 
If one side line of a billiard in an ellipsoid $\mathcal E$ is located on a focal axis, then the billiard is a focal billiard.
This means, all side lines are generators of a confocal one-sheeted hyperboloid $\mathcal H_1$, and the vertices of the billiard are located on the line of curvature $e = \mathcal E \cap \mathcal H_1$ (\Figref{fig:raum_bil}).
\end{lem}

\begin{proof}
This follows from the symmetry of the two $\mathcal H_1$-generators through any point $P\in e$ w.r.t.\ the tangent plane to $\mathcal E$ at $P$.
Along both loops of the curve of intersection $e = \mathcal E \cap \mathcal H_1$, the elliptic coordinates $k_0$ and $k_1$ remain constant, while $k_2$ varies.
\end{proof}

\subsection{Focal billiards and Henrici's flexible hyperboloid}

\begin{thm}\label{thm:isometric2} 
Referring to the notation in \Thmref{thm:isometric}, suppose that the billiard $P_1' P_2' \dots$ lies in the $[x,y]$-plane and the isometric billiard $P_1'' P_2'' \dots$ in the $[x,z]$ plane, such that the circumscribed ellipses $e'$ and $e''$ share the principal vertices on the $x$-axis and the semiaxes satisfy \eqref{eq:iso}.
Then the isometric transition from $P_1' P_2' \dots$ to $P_1'' P_2'' \dots$ can be carried out continuously via mutually isometric focal billiards $P_1 P_2 \dots$ in the fixed triaxial ellipsoid $\mathcal E$ through $e'$ and $e''$ with the focal ellipse $c'$ and the focal hyperbola $c''$ (\Figref{fig:raum_bil}).
\\[1.0mm]
During this transition, which is an affine motion, the paths of the billiards' vertices on $\mathcal E$ are orthogonal trajectories of the confocal one-sheeted hyperboloids.
In the two flat poses, the sides of the billiards $P_1' P_2' \dots$ in $e'$ and $P_1'' P_2'' \dots$ in $e''$ are tangent to the respectively coplanar focal conics $c'$ and $c''$.
\end{thm}

\begin{figure}[htb] 
  \centering
  \def\sz{\small} 
  \psfrag{P1}[ct]{\contournumber{32}\contourlength{1.6pt}\contour{hgrau}{\sz\red $P_1$}}
  \psfrag{P2}[lb]{\contournumber{32}\contourlength{1.6pt}\contour{hgrau}{\sz\red $P_2$}}
  \psfrag{P3}[rb]{\contourlength{1.2pt}\contour{white}{\sz\red $P_3$}}
  \psfrag{P4}[rt]{\contourlength{1.2pt}\contour{white}{\sz\red $P_4$}}
  \psfrag{P5}[ct]{\contourlength{1.2pt}\contour{white}{\sz\red $P_5$}}
  \psfrag{P6}[lt]{\contournumber{32}\contourlength{1.4pt}\contour{white}{\sz\red $P_6$}}
  \psfrag{P7}[lb]{\contourlength{1.2pt}\contour{white}{\sz\red $P_7$}}
  \psfrag{P8}[cb]{\contourlength{1.2pt}\contour{dgrau}{\sz\white $P_8$}}
  \psfrag{P9}[rt]{\contournumber{32}\contourlength{1.6pt}\contour{hgrau}{\sz\red $P_9$}}
  \psfrag{P10}[lt]{\contourlength{1.2pt}\contour{white}{\sz\red $P_{10}$}}
  \psfrag{P11}[lb]{\contourlength{0.5pt}\contour{white}{\sz\red $P_{11}$}}
  \psfrag{P12}[cb]{\contourlength{1.2pt}\contour{dgrau}{\sz\white $P_{12}$}}
  \psfrag{P13}[rb]{\contournumber{32}\contourlength{1.3pt}\contour{white}{\sz\red $P_{13}$}}
  \psfrag{P14}[rt]{\contourlength{1.2pt}\contour{white}{\sz\red $P_{14}$}}
  \psfrag{P1'}[lt]{\contournumber{32}\contourlength{1.6pt}\contour{hgrau}{\sz $P_1'$}}
  \psfrag{P3'}[lt]{\contourlength{1.2pt}\contour{dgrau}{\sz\white $P_3'$}}
  \psfrag{P5'}[ct]{\contournumber{32}\contourlength{1.6pt}\contour{hgrau}{\sz $P_5'$}}
  \psfrag{P8'}[lc]{\contourlength{1.2pt}\contour{dgrau}{\sz\white $P_8'$}}
  \psfrag{P12'}[cb]{\contournumber{32}\contourlength{1.6pt}\contour{dgrau}{\sz\white $P_{12}'$}}  
  \psfrag{P14'}[lb]{\contournumber{32}\contourlength{1.4pt}\contour{hgrau}{\sz $P_{14}'$}}
  \psfrag{E}[lt]{\blue $\mathcal E$}
  \psfrag{e}[rb]{\sz\red $\boldsymbol{e}$}
  \psfrag{e'o}[rb]{\sz\white $\boldsymbol{e'}$}  
  \psfrag{e'u}[rb]{\contourlength{0.0pt}\contour{white}{\sz $\boldsymbol{e'}$}}
  \psfrag{c}[lt]{\sz\white $\boldsymbol{c'}$} 
  \psfrag{x}[lt]{\sz $x$}
  \psfrag{y}[rt]{\sz $y$}
  \includegraphics[width=130mm]{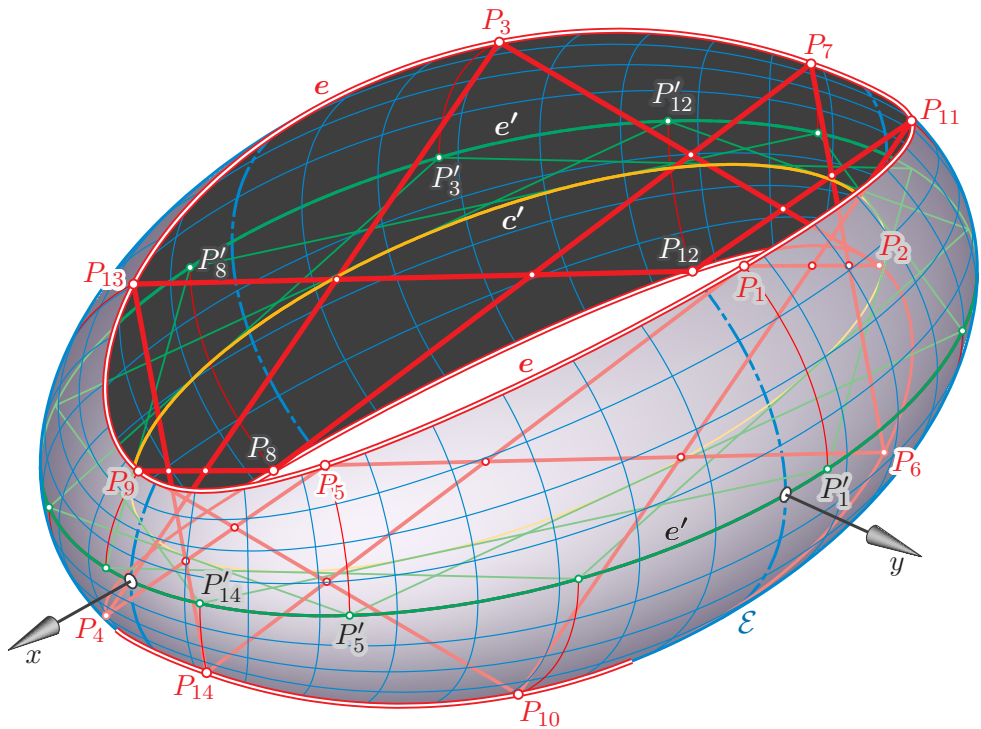} 
  \caption{The transition from the periodic flat billiard $P_1'P_2'\dots P_{14}'$ (green) in $e'$ with caustic $c'$ to the isometric focal billiard $P_1P_2\dots P_{14}$ (red) in the ellipsoid $\mathcal E$ with vertices on the line of curvature $e$.}
  \label{fig:raum_bil}
\end{figure}

In order to prove \Thmref{thm:isometric2}, we focus on any pair of confocal one-sheeted hyperboloids $\mathcal H_1$ and $\mathcal H_1\iso$ with the respective elliptic coordinates $k_1$ and $k_1\iso$ and $-b_c\strqu < k_1,k_1\iso < 0$.
There is an axial scaling
\begin{equation}\label{eq:corr_points} 
 \begin{array}{rl}
   \gamma(k_1,k_1\iso)\!:  
    &P = (x,y,z)\, \mapsto \, P\iso = (x\iso, y\iso, z\iso)\zwi\mbox{with}
     \zwi \mathcal H_1 \to \mathcal H_1\iso \zwi\mbox{and}
   \\[3.0mm]
    &x\iso = x\, \sqrt{\Frac{a_c\strqu + k_1\iso}{a_c\strqu + k_1}}\,, \quad 
     y\iso = y\, \sqrt{\Frac{b_c\strqu + k_1\iso}{b_c\strqu + k_1}}\,, \quad 
     z\iso = z\,\sqrt{\Frac{k_1\iso}{k_1}}\,.
 \end{array}
\end{equation}
This affine transformation maps generators of $\mathcal H_1$ to those of $\mathcal H_1\iso$, but also lines of curvature of $\mathcal H_1$ to those of $\mathcal H_1\iso$.
The latter follows from the Lemma below.

\begin{lem}\label{lem:corr_points} 
Let $P\in\mathcal H_1$ and $P\iso\in\mathcal H_1\iso$ be corresponding points under the affine transformation $\gamma(k_1,k_1\iso)$ defined in {\em \eqref{eq:corr_points}}.
Then, $P$ and $P\iso$ share the elliptic coordinates $k_0$ and $k_2$.
\end{lem} 

\begin{proof}
It is sufficient to show that the ellipsoid and the two-sheeted hyperboloid through $P = (x,y,z)$ passes also through $P\iso = (x\iso, y\iso, z\iso)$.
For the polynomial
\[  F(P,k):= \frac{x^2}{a_c\strqu + k} + \frac{y^2}{b_c\strqu + k} + \frac{z^2}{k} - 1
\]
holds $F(P,\,k_1) =  0$ and $F(P\iso,\,k_1\iso) = 0$.
We prove for $k = k_j$, $j = 0,2$ and therefore $k_j\ne k_1$, that the equation $F(P,\,k_j) = 0$ implies $F(P\iso,\,k_j) = 0$:
\\
The claim is a consequence of the identity
\[  F(P\iso,\,k_j) = \smFrac{k_1\iso-k_1}{k_j-k_1}\,F(P,k_1) + \left(1 - \smFrac{k_1\iso-k_1}{k_j-k_1}\right)F(P,\,k_j).
\]
This representation as an affine combination can readily be verified by 
\[  \frac{1}{a_c\strqu + k_j}\,x\strqu = 
    \frac{a_c\strqu + k_1\iso}{(a_c\strqu + k_1)(a_c\strqu + k_j)}\,x^2 = 
    \left[\frac{k_1\iso - k_1}{(k_j - k_1)(a_c\strqu + k_1)} 
     + \frac{k_j - k_1\iso}{(k_j - k_1)(a_c\strqu + k_j)}\right] x^2. 
\]
Similar identities hold for the terms with $y^2$ and $z^2$ in $F(P\iso,\,k_j)$.
\end{proof}

The choice $k_1\iso = 0$ in \eqref{eq:corr_points} defines a singular affine transformation which sends the focal billiard $P_1P_2\dots$ to the isometric planar billiard $P_1'P_2'\dots$ inscribed in the ellipse $e'\subset\mathcal E$ in the plane $z=0$ (\Figref{fig:raum_bil}).
From the standard parametrization $x' = a_e'\cos t$, $y' = b_e'\sin t$ of $e'$ follows with \eqref{eq:corr_points} for the curve $e = \mathcal E\cap\mathcal H_1$
\[  x(t) = \frac{a_{h_1}}{a_c'}\,x'\,,\quad
    y(t) = \frac{b_{h_1}}{b_c'}\,y'\,
\]
with 
\begin{equation}\label{eq:ah,bh,ch}
  a_{h_1} = \sqrt{a_c\strqu + k_1}, \quad b_{h_1} = \sqrt{b_c\strqu + k_1},\quad
    c_{h_1} = \sqrt{-k_1}
\end{equation}
as semiaxes of $\mathcal H_1$.
This yields for the upper and the lower loop of $e$ the parametrizations
\begin{equation}\label{eq:Bahn_Raumbil} 
  \Vkt e(t) = \left( x(t),\,y(t),\,z(t)\right) = \left(
   \Frac{a_e a_{h_1}}{a_c'}\,\cos t, \ \Frac{b_e b_{h_1}}{b_c'}\,\sin t, \
   \pm\Frac{\sqrt{k_0 k_1 k_2(t)}}{a_c' b_c'} \right), 
\end{equation}
where the third equation follows from \eqref{eq:elliptic-cartesian} with $k_0 = k_e'$ and $k_2 = k_h'$ by \eqref{eq:k_h}, i.e., for
\begin{equation}\label{eq:k0,k2}
   k_0 = a_e^2 - a_c\strqu = b_e^2 - b_c\strqu = c_e^2, \quad
   k_2 = k_2(t) = -(a_c\strqu \sin^2 t + b_c\strqu \cos^2 t).
\end{equation}

For the second limit $k_1\iso = -b_c\strqu$, the image under $\gamma(k_1,-b_c\strqu)$ is the isometric planar billiard $P_1'' P_2''\dots$ (\Figref{fig:isometric2}) inscribed in the ellipse $e''\subset\mathcal E$ in $y=0$ with the focal hyperbola $c''$ of $\mathcal E$ as caustic.
By virtue of \eqref{eq:corr_points}, we obtain for the upper and lower subarcs of $e''$ the parametrizations
\begin{equation}\label{eq:param hyp_bil}
   x''(t) = x(t) \sqrt{\smFrac{a_c\strqu - b_c\strqu}{a_c\strqu + k_1}}
          = \frac{a_e d' \cos t}{a_c'}\,, \quad y'' = 0,\quad
   z''(t) = z(t) \frac{b_c'}{\sqrt{-k_1}} 
         = \pm \frac{\sqrt{-k_0 k_2(t)}}{a_c'}
\end{equation}
The extreme locations of the vertices $P_i''\in e''$ with $t=0$ or $\pi$ are umbilics of the ellipsoid $\mathcal E$ (see \cite[Fig.~2]{Quadrics}).

\begin{rem} 
For any two poses $P_1P_2\dots$ and $P_1\iso P_2\iso\dots$ including the flat limits $P_1'P_2'\dots$ and $P_1'' P_2''\dots$, the corresponding pairs of consecutive vertices
form a skew isogram $P_iP_{i+1}P\iso_iP\iso_{i+1}$, since $\ol{P_iP_{i+1}} = 
\ol{P\iso_i P\iso _{i+1}}$ and $\ol{P_iP\iso_{i+1}} = \ol{P\iso_i P_{i+1}}$.
The latter equation is a consequence of Ivory's Theorem (see, e.g., \cite[p.~306]{Quadrics}).
Skew isograms have an axis of symmetry which connects the midpoints of the two diagonals.
Therefore, the isometry between corresponding side lines $[P_i, P_{i+1}]$ and $[P\iso_i, P\iso_{i+1}]$ is an axial symmetry. 
\\
By the same token, Ivory's theorem implies more general $\ol{P_iP\iso_j} = \ol{P\iso_i P_j}$ for all $(i,j)$. 
\end{rem}

\begin{figure}[htb] 
  \centering
  \def\sz{\small} 
  \psfrag{E}[cc]{\blue $\mathcal E$}
  \psfrag{H}[cc]{\blue $\mathcal H_1$}
  \psfrag{e1}[rb]{\sz\red $\boldsymbol{e}$}
  \psfrag{1}[lb]{\contourlength{1.2pt}\contour{hgrau}{\sz $P_1$}}
  \psfrag{2}[rb]{\contourlength{1.2pt}\contour{white}{\sz $P_2$}}
  \psfrag{6}[lt]{\contourlength{1.2pt}\contour{white}{\sz $P_6$}}
  \psfrag{7}[rb]{\contourlength{1.2pt}\contour{hgrau}{\sz $P_7$}}
  \psfrag{10}[rt]{\contournumber{32}\contourlength{1.2pt}\contour{white}{\sz $P_{10}$}}
  \psfrag{11}[rb]{\contourlength{1.2pt}\contour{hgrau}{\sz $P_{11}$}}
  \psfrag{14}[ct]{\contournumber{32}\contourlength{1.2pt}\contour{white}{\sz $P_{14}$}}
  \psfrag{13}[rb]{\contourlength{1.4pt}\contour{dgrau}{\sz\white $P_{13}$}}
  \psfrag{15}[rb]{\contourlength{1.2pt}\contour{hgrau}{\sz $P_{15}$}}
  \psfrag{16}[lb]{\contournumber{32}\contourlength{1.2pt}\contour{gelbb}{\sz $P_{16}$}}
  \psfrag{19}[rb]{\contourlength{1.2pt}\contour{hgrau}{\sz $P_{19}$}}
  \psfrag{20}[rb]{\contourlength{1.2pt}\contour{white}{\sz $P_{20}$}}
   \psfrag{S3^2}[rt]{\sz\red $S_3^{(2)}$}
   \psfrag{S4^2}[lt]{\sz\red $S_4^{(2)}$}
   \psfrag{S7^2}[rc]{\sz\red $S_7^{(2)}$}
   \psfrag{S8^2}[rt]{\sz\red $S_8^{(2)}$}
   \psfrag{S9^2}[lb]{\sz\red $S_9^{(2)}$}
   \psfrag{S11^2}[rb]{\sz\red $S_{11}^{(2)}$}
   \psfrag{S12^2}[ct]{\contourlength{1.2pt}\contour{white}{\sz\red $S_{12}^{(2)}$}}
   \psfrag{S13^2}[lt]{\contourlength{1.2pt}\contour{hgrau}{\sz\red $S_{13}^{(2)}$}}
   \psfrag{S15^2}[rb]{\sz\red $S_{15}^{(2)}$}
   \psfrag{S16^2}[rc]{\sz\red $S_{16}^{(2)}$}
   \psfrag{S17^2}[lt]{\contourlength{1.2pt}\contour{hgrau}{\sz\red $S_{17}^{(2)}$}}
   \psfrag{S18^2}[lc]{\sz\red $S_{18}^{(2)}$}
   \psfrag{S19^2}[lb]{\sz\red $S_{19}^{(2)}$}
   \psfrag{S21^2}[lt]{\contourlength{1.2pt}\contour{hgrau}{\sz\red $S_{21}^{(2)}$}}
   \psfrag{S22^2}[lt]{\sz\red $S_{22}^{(2)}$}
   \psfrag{Q1}[lb]{\contourlength{1.2pt}\contour{hgrau}{\sz\blue $Q_1$}}
   \psfrag{S6^8}[ct]{\contourlength{1.2pt}\contour{hgrau}{\sz\red $S_{6}^{(8)}$}}
   \psfrag{S4^4}[ct]{\contourlength{1.2pt}\contour{hgrau}{\sz $S_{4}^{(4)}$}}
   \psfrag{S21^4}[lc]{\contourlength{1.2pt}\contour{hgrau}{\sz\red $S_{21}^{(4)}$}}

   \psfrag{e}[rb]{\gelbb $\boldsymbol{e}$}
   \psfrag{e2}[lt]{\red $\boldsymbol{e}^{(2)}$}
   \psfrag{e4}[ct]{$\boldsymbol{e}^{(4)}$}
   \psfrag{e8}[ct]{$\boldsymbol{e}^{(8)}$}
   \psfrag{E}[rc]{\blue $\mathcal E$}
   \psfrag{H1}[lt]{\blue $\mathcal H_1$}
   \psfrag{H1u}[cc]{\blue $\mathcal H_1$}
  \includegraphics[width=110mm]{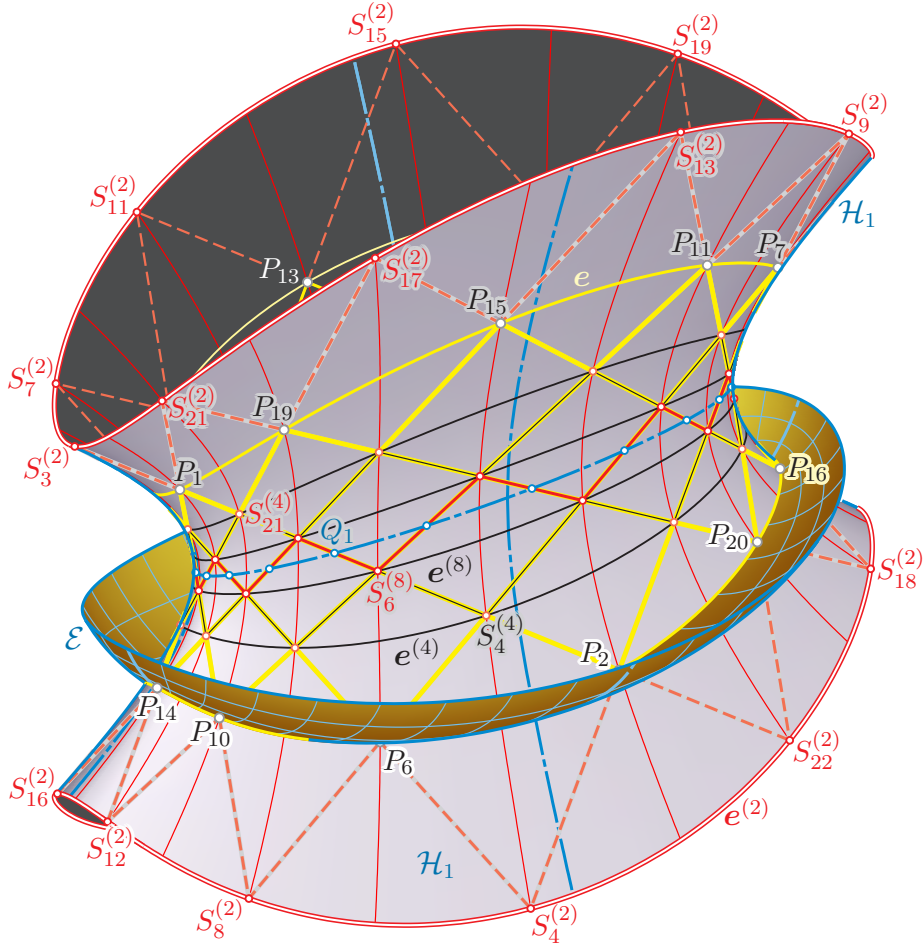} 
  \caption{Periodic focal billiard in $\mathcal E$ along $e$ (yellow) with $N=22$ and turning number $\tau=5$ together with the spatial Poncelet grid.
The extended sides determine focal billiards along $e^{(2)}$ with $\tau = 7$ (red dashed), along $e^{(4)}$ with $\tau=3$ (black), and along $e^{(8)}$ with $\tau=1$ (red).}
  \label{fig:bil_mit_Hyp} 
\end{figure}

\begin{figure}[htb] 
  \centering
  \def\sz{\small} 
  \psfrag{E}[cc]{\blue $\mathcal E$}
  \psfrag{H}[cc]{\blue $\mathcal H_1$}
  \psfrag{e1}[rb]{\sz\red $\boldsymbol{e}$}
  \psfrag{1}[lb]{\contourlength{1.2pt}\contour{hgrau}{\sz $P_1$}}
  \psfrag{2}[rb]{}
  \psfrag{3}[rb]{}
  \psfrag{4}[rb]{}
  \psfrag{5}[rb]{}
  \psfrag{6}[lt]{\contourlength{1.2pt}\contour{white}{\sz $P_6$}}
  \psfrag{7}[lt]{\sz $P_7$}
  \psfrag{8}[rb]{}  
  \psfrag{9}[lc]{\contourlength{1.2pt}\contour{dgrau}{\sz\white $P_9$}}
  \psfrag{10}[rt]{\contournumber{32}\contourlength{1.2pt}\contour{white}{\sz $P_{10}$}}
  \psfrag{11}[rb]{\contourlength{1.2pt}\contour{hgrau}{\sz $P_{11}$}}
  \psfrag{12}[rb]{}  
  \psfrag{14}[ct]{\contournumber{32}\contourlength{1.2pt}\contour{white}{\sz $P_{14}$}}
  \psfrag{13}[cb]{\contourlength{1.4pt}\contour{dgrau}{\sz\white $P_{13}$}}
  \psfrag{15}[rb]{\contourlength{1.2pt}\contour{hgrau}{\sz $P_{15}$}}
  \psfrag{16}[lb]{\contournumber{32}\contourlength{1.2pt}\contour{gelbb}{\sz $P_{16}$}}
  \psfrag{17}[cb]{\contourlength{1.2pt}\contour{dgrau}{\sz\white $P_{17}$}}
  \psfrag{18}[rb]{}
  \psfrag{19}[cb]{\contourlength{1.2pt}\contour{hgrau}{\sz $P_{19}$}}
  \psfrag{20}[rb]{\contourlength{1.2pt}\contour{hgrau}{\sz $P_{20}$}}
  \psfrag{21}[rb]{}
  \psfrag{22}[rb]{}  
   \psfrag{S3^2}[rt]{\sz\red $S_3^{(2)}$}
   \psfrag{S4^2}[lt]{\sz\red $S_4^{(2)}$}
   \psfrag{S7^2}[rc]{\sz\red $S_7^{(2)}$}
   \psfrag{S8^2}[rt]{\sz\red $S_8^{(2)}$}
   \psfrag{S9^2}[lb]{\sz\red $S_9^{(2)}$}
   \psfrag{S11^2}[rb]{\sz\red $S_{11}^{(2)}$}
   \psfrag{S12^2}[ct]{\contourlength{1.2pt}\contour{white}{\sz\red $S_{12}^{(2)}$}}
   \psfrag{S13^2}[lt]{\contourlength{1.2pt}\contour{hgrau}{\sz\red $S_{13}^{(2)}$}}
   \psfrag{S15^2}[rb]{\sz\red $S_{15}^{(2)}$}
   \psfrag{S16^2}[rc]{\sz\red $S_{16}^{(2)}$}
   \psfrag{S17^2}[lt]{\contourlength{1.2pt}\contour{hgrau}{\sz\red $S_{17}^{(2)}$}}
   \psfrag{S18^2}[lc]{\sz\red $S_{18}^{(2)}$}
   \psfrag{S19^2}[lb]{\sz\red $S_{19}^{(2)}$}
   \psfrag{S21^2}[lt]{\contourlength{1.2pt}\contour{hgrau}{\sz\red $S_{21}^{(2)}$}}
   \psfrag{S22^2}[lt]{\sz\red $S_{22}^{(2)}$}
   \psfrag{Q1}[lb]{\contourlength{1.2pt}\contour{hgrau}{\sz\blue $Q_1$}}
   \psfrag{S6^8}[ct]{\contourlength{1.2pt}\contour{hgrau}{\sz\red $S_{6}^{(8)}$}}
   \psfrag{S4^4}[ct]{\contourlength{1.2pt}\contour{hgrau}{\sz $S_{4}^{(4)}$}}
   \psfrag{S21^4}[lc]{\contourlength{1.2pt}\contour{hgrau}{\sz\red $S_{21}^{(4)}$}}
   \psfrag{e}[rb]{\gelbb $\boldsymbol{e}$}
   \psfrag{e2}[lt]{\red $\boldsymbol{e}^{(2)}$}
   \psfrag{e4}[ct]{$\boldsymbol{e}^{(4)}$}
   \psfrag{e8}[ct]{$\boldsymbol{e}^{(8)}$}
   \psfrag{E}[rc]{\blue $\mathcal E$}
   \psfrag{H1}[lt]{\blue $\mathcal H_1$}
   \psfrag{H1u}[cc]{\blue $\mathcal H_1$}
  \includegraphics[width=120mm]{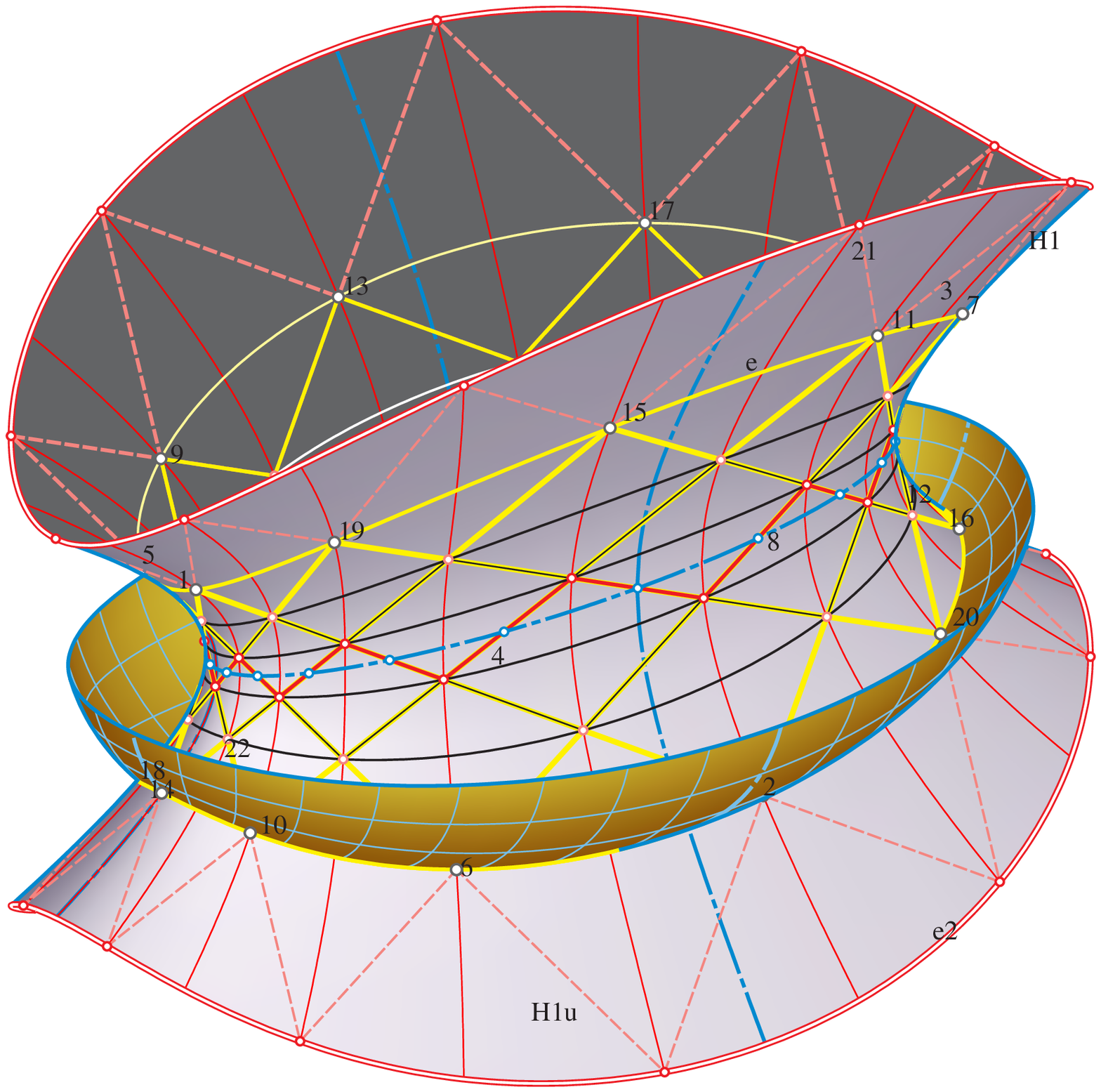} 
  \caption{Another pose of a periodic focal billiard in $\mathcal E$ along $e$ (yellow) with $N=22$ and turning number $\tau=5$ together with the spatial Poncelet grid.}
  \label{fig:bil_mit_Hyp2} 
\end{figure}

\begin{proof} (\Thmref{thm:isometric2})  
According to \eqref{eq:Bahn_Raumbil}, the coordinates of the billiards' vertices are continuous functions of $k_1$.
The variation of $k_1$ from $0$ to $-b_c\strqu$ shows the stated transition from $P_i'$ via $P_i$ to $P_i''$ for all $i\,$.
All trajectories are $k_1$-lines and therefore orthogonal to the confocal one-sheeted hyperboloids $k_1 = \mbox{const.}$ through the billiards' sides.
Hence for each side $P_iP_{i+1}$, the tangents to the trajectories of the endpoints are orthogonal to the side.
Consequently, the length of each side remains constant during this transition.
Moreover, each point attached to the moving side line $[P_i,P_{i+1}]$ remains on the same ellipsoid and two-sheeted hyperboloid in the confocal family, since only $k_1$ varies. 
\\
In particular, the contact point $Q_i'$ of the side $P_i'P_{i+1}'$ with the caustic $c'$ remains in the plane  $z=0$ and runs along a hyperbola. 
Similarly, the intersection point $T_i'$ of the side $P_i'P_{i+1}'$ with the $x$-axis remains in the plane $y=0$ and runs along an ellipse.
The end pose of this point is the contact point $T_i''$ of the side line $[P_i'',P_{i+1}'']$ with the caustic $c''$ (\Figref{fig:isometric2}).
\end{proof}

\begin{rem}
In each pose, the sides of the billiards can serve as rods of a flexible {\em Henrici hyperboloid} with spherical joints at all vertices and crossing points between sides (see \cite[p.~88]{Quadrics} or \cite[Fig.~17]{Barrow}).
Different to the standard model of this framework \cite{Henrici}, the endpoints of the rods are not located in planar sections but on closed lines of curvature on the confocal one-sheeted hyperboloid. 
If the axes of symmetry are fixed and the rods of the framework represent a periodic focal billiard, then during the flexion each crossing point runs along a line of curvature on a confocal ellipsoid or, in particular, on a hyperbola in the plane of the gorge ellipse or on an ellipse in the plane of the focal hyperbola.
\end{rem}

We conclude this section with a few consequences of \Lemref{lem:corr_points} and  
\Thmref{thm:isometric2} for any periodic focal billiard in the ellipsoid $\mathcal E$ and its flat limits with the focal ellipse $c'$ and the focal hyperbola $c''$ of $\mathcal E$ as respective caustics.

\begin{cor}\label{cor:L} 
If there exist periodic focal billiards in the triaxial ellipsoid $\mathcal E$, then all focal billiards in $\mathcal E$ share the perimeter $L_{\mathcal E}\,$.  
\end{cor}

\begin{lem}\label{lem:corr_points_ell_coord} 
If the planar billiard in $e'$ with the caustic $c'$ is $N$-periodic with turning number $\tau$, then all focal billiards in the ellipsoid $\mathcal E$ through $e'$ with the focal ellipse $c'$ are 
\begin{itemize}
\item for even $N$, again $N$-periodic with turning number $\tau$,
\item for odd $N$, $2N$-periodic with turning number $2\tau$.
\end{itemize}
\end{lem} 

\begin{lem}\label{lem:spat_bil_symm} 
$N$-periodic focal billiards in an triaxial ellipsoid $\mathcal E$ are

\begin{enumerate}
\item[\em (i)] for $N\equiv 2\pmod 4$ and even $\tau$, symmetric w.r.t.\ the plane $z=0$ of the focal ellipse,
\item[\em (ii)] for $N\equiv 2\pmod 4$ and odd $\tau$, symmetric w.r.t.\ the center $O$ of $\mathcal E$,
\item[\em (iii)] for $N\equiv 0\pmod 4$ and odd $\tau$, axial symmetric w.r.t.\ the $z$-axis.
Focal billiards with $N\equiv 0\pmod 4$ and even $\tau$ split.
\end{enumerate}
\end{lem} 

\begin{proof}
This follows from the relative position of opposite vertices $P_i$ and $P_{i+N/2}\,$.
\end{proof}

By virtue of \Lemref{lem:spat_bil_symm}, all periodic focal billiards in ellipsoids admit a symmetry which interchanges opposite vertices.
This has consequences for the flat limits.
Below, we use the symbol $\mathrm{ar}\,(A,B,C) = \ol{AC}/\ol{BC}$ (distances signed) for the {\em affine ratio} of three collinear points $A,B,C$.

\begin{cor}\label{cor:li/ri} 
\begin{enumerate}
\item If $P_1'P_2'\dots P_N'$ is an $N$-periodic billiard in the ellipse $e'$ with $Q_1', Q_2', \dots$ as contact points with the ellipse $c'$ as caustic, then 
\[  \mbox{for even $N$} \ 
    \prod_{i=1}^{N/2} \, \mathrm{ar}\,(P_i',P_{i+1}',Q_i') = (-1)^{N/2}, \ 
    \mbox{for odd $N$}\ \prod_{i=1}^N \, \mathrm{ar}\,(P_i',P_{i+1}',Q_i') = 1\,.
\]
\item If $P_1''P_2''\dots P_N''$ is an $N$-periodic billiard with turning number $\tau$ in an ellipse $e''$ with $T_1'', T_2'', \dots$ as contact points with the hyperbola $c''$ as caustic and $P_i''\ne T_i''$ for all $i$, then 
\[  \prod_{i=1}^{N/2} \, \mathrm{ar}\,(P_i,P_{i+1},T_i) = (-1)^\tau.
\]
\end{enumerate}
\end{cor}

\begin{proof}
Due to \Thmref{thm:isometric2}, we prove the statements for an isometric focal billiard $P_1P_2\dots P_N$ in an ellipsoid $\mathcal E$ in standard position, where $Q_i$ and $T_i$ are the intersection points of the sides $[P_i,P_{i+1}]$ with $z=0$ and $y=0$, respectively.
Hence 
\[  \mathrm{ar}\,(P_i',P_{i+1}',Q_i') = \mathrm{ar}\,(P_i,P_{i+1},Q_i) 
    = \frac{z_i}{z_{i+1}}, \quad 
    \mathrm{ar}\,(P_i'',P_{i+1}'',T_i'') = \mathrm{ar}\,(P_i,P_{i+1},T_i) 
    = \frac{y_i}{y_{i+1}}, 
\]
where $P_j = (x_j,y_j,z_j)$ for all $j$.
Thus, we obtain for the two products under consideration $z_1/z_{N/2}$ and $y_1/y_{N/2}$.
The claims follow from the symmetries as listed in \Lemref{lem:spat_bil_symm}.
\end{proof}

\section{The geometry of focal billiards in ellipsoids}

\subsection{Side lengths and angles}

\begin{lem}\label{lem:wkl_raum} 
If the vertex $P\in e = \mathcal E\cap\mathcal H_1$ of the focal billiard has the elliptic coordinates $(k_0,k_1,k_2)$, then for the angle $\theta/2$ between the adjacent sides and the tangential plane to $\mathcal E$ at $P$ holds (note \Figref{fig:symm_Erz}) 
\begin{equation}\label{eq:wkl_raum}
  \tan\frac{\theta}2 = \pm \sqrt{\frac{k_0 - k_1}{k_1 - k_2}}\,,\quad
  \sin\frac{\theta}2 = \pm \sqrt{\frac{k_0 - k_1}{k_0 - k_2}}\,.
\end{equation}
\end{lem}

\begin{proof}
The principal curvature directions of $\mathcal H_1$ at $P$ are normal to the ellipsoid $\mathcal E$ and the two-sheeted hyperboloid $\mathcal H_2$ through $P$.
The poles of the tangent plane $\tau_P$ to $\mathcal H_1$ w.r.t.\ $\mathcal E$ and $\mathcal H_2$ are the centers of curvature of orthogonal sections of $\mathcal H_1$ through the principal curvature directions \cite[Lemma~8.2.1]{Quadrics}.
Therefore, the respective radii of curvature are proportional to $k_0-k_1$ and $k_1-k_2$.
\\
As a consequence, at the point $P$ the two generators of $\mathcal H_1$ have direction vectors
\begin{equation}\label{eq:Erzeugende}
  \sqrt{k_1-k_2}\ \frac{\Vkt n_2}{\Vert\Vkt n_2\Vert} 
    \pm \sqrt{k_0-k_1}\ \frac{\Vkt n_0}{\Vert\Vkt n_0\Vert}\,,
\end{equation}
when $\Vkt n_0$ is orthogonal to $\mathcal E$ and $\Vkt n_2$ orthogonal to $\mathcal H_2$.
This completes the proof. 
\end{proof}

\medskip
\begin{rem}
The curves of constant angle $\theta$ on $\mathcal H_1$ are the curves of constant ratio
of the elliptic coordinates relative to $\mathcal H_1$ or of the main curvatures.
These curves played a role in \cite{Sta_motion}.
The particular case $\theta=\pi/2$ yields the intersection curve with the director sphere
with radius $\sqrt{a_{h_1}^2 + b_{h_1}^2 - c_{h_1}^2}$ (see \cite{Wunderlich} or \cite[p.~46]{Quadrics}).  
\end{rem}

\medskip
At the limit $k_1 = 0$, \eqref{eq:wkl_raum} agrees with the formula
\[ \tan\frac{\theta(t)}2 = \pm \frac{\sqrt{k_0}}{\Vert \Vkt t_c(t)\Vert}
   = \pm \sqrt{\frac{k_0}{-k_2(t)}} 
\]   
in \cite[Lemma~2.2]{Sta_I}.
For $k_1 = -b_c\strqu$, i.e., for the hyperbolic billiard in the ellipse $e''$ follows 
\[  \tan\frac{\theta(t)}2 = \pm \sqrt{\frac{k_0 + b_c\strqu}{-b_c\strqu - k_2(t)}}
    = \pm \frac{b_e'}{d' \sin t}
     = \pm \sqrt{\frac{k_e''}{a_c\zwstrqu \sin^2 t'' - b_c\zwstrqu \cos^2 t''}}\,,
\]
where $t''$ denotes the standard parameter of the ellipse $e''$ with semiaxes $a_e'' = a_e$ and $b_e'' = c_e = k_e'\,$, i.e., by \eqref{eq:param hyp_bil}
\[  \cos t'' = \frac{d'}{a_c'}\,\cos t,   \quad 
    \sin t'' = \frac{\sqrt{a_c\strqu \sin^2 t + b_c\strqu \cos^2 t}}{a_c'}\,. 
\]    

According to \eqref{eq:n1n2n3}, the normal vector of the ellipsoid along the curve $e = \mathcal E\cap \mathcal H_1$ with the parametrization $\Vkt e(t)$ from \eqref{eq:Bahn_Raumbil} is
\[  \Vkt n_0 = \left( \frac{x(t)}{a_e^2}\,, \ \frac{y(t)}{b_e^2}\,, \ 
    \frac{z(t)}{c_e^2}\right) =
    \left( \frac{a_{h_1}}{a_c' a_e}\cos t, \ \frac{b_{h_1}}{b_c' b_e}\sin t, \ 
    \frac{\sqrt{k_1k_2}}{a_c' b_c' \sqrt{k_0}}\right).
\]
This implies with \eqref{eq:k0,k2}
  \def\arraycolsep{1.0mm}
\[ \begin{array}{rcl}
     \Vert \Vkt n_0 \Vert^2 &= &\Frac{a_c\strqu + k_1}{a_c\strqu a_e^2}\cos^2 t
    + \Frac{b_c\strqu + k_1}{b_c\strqu b_e^2}\sin^2 t 
    - \Frac{k_1(a_c\strqu \sin^2 t + b_c\strqu \cos^2 t)}{a_c\strqu b_c\strqu k_0}
   \\[4.0mm]
   &= & \Frac{(b_c\strqu \cos^2 t + a_c\strqu \sin^2 t + k_0)(k_0 - k_1)}{a_e^2 b_e^2 k_0}
    = \Frac{(k_0-k_1)(k_0-k_2)}{a_e^2 b_e^2 k_0}\,.
  \end{array}
\]
Similar computations for the other normal vectors result in
\begin{equation}\label{eq:gradients}
 \begin{array}{c}
   \Vert\Vkt n_0\Vert = \Frac{\sqrt{(k_0-k_1)(k_0-k_2)}}{a_e b_e c_e}\,, \quad
   \Vert\Vkt n_1\Vert = \Frac{\sqrt{(k_0-k_1)(k_1-k_2)}}{a_{h_1} b_{h_1} c_{h_1}}\,, 
   \\[3.0mm]
   \Vert\Vkt n_2\Vert = \Frac{\sqrt{(k_0-k_2)(k_1-k_2)}}
    {a_{h_2} b_{h_2} c_{h_2}}\,, 
 \end{array}
\end{equation}
where
\begin{equation}\label{eq:ah2_bh2_ch2}
  a_{h_2} = \sqrt{a_c\strqu + k_2},\quad b_{h_2} = \sqrt{-(b_c\strqu + k_2)},\quad
    c_{h_2} = \sqrt{-k_2}
\end{equation}
are the semiaxes of the two-sheeted hyperboloid $\mathcal H_2$ which passes through the point with position vector $\Vkt e(t)$.
By virtue of \eqref{eq:gradients}, the direction vectors \eqref{eq:Erzeugende} of the generators on $\mathcal H_1$ can be expressed as
\begin{equation}\label{eq:Erzeugende2}
  a_e b_e c_e\,\Vkt n_0 \pm a_{h_2} b_{h_2} c_{h_2}\,\Vkt n_2\,. 
\end{equation}

\medskip
Based on the orientation of the normalvector $\Vkt n_1$ of $\mathcal H_1$ in \eqref{eq:n1n2n3} pointing outwards, we can also orientate the angles $\theta_i$ along $e$.
We stick to a counter-clockwise order of the initial billiard in $e'$ and notice that the angles $\theta_i$ change their sign from vertex to vertex in \eqref{eq:wkl_raum}.

The tangent vector to $e = \mathcal E \cap\mathcal H_1$ from \eqref{eq:Bahn_Raumbil} is
\[  \Vkt t_e(t):= \frac{\mathrm d}{\mathrm d t}\Vkt e(t) = \left(
    \frac{-a_e a_{h_1}\sin t}{a_c'}, \ \frac{b_e b_{h_1}\cos t}{b_c'}, \
    \pm\frac{\sqrt{-k_0k_1}\,\sqrt{-(a_c\strqu + k_2)(b_c\strqu + k_2)}}{a_c' b_c' \sqrt{-k_2}} \right).      
\]
This shows that
\[  \Vkt t_e(t) = \lambda \Vkt n_2 \quad\mbox{with}\quad
    \lambda = -d\strqu \sin t \cos t = \sqrt{-(a\strqu + k_2)(b\strqu + k_2)} 
    = a_{h_2} b_{h_2}\,.  
\]
We find \def\arraycolsep{0.8mm}
\[ \begin{array}{rcl}
     \Vert \Vkt t_e(t)\Vert^2
    &= &\Frac{a_e^2(a_c\strqu + k_1)\sin^2 t}{a_c\strqu} 
       + \Frac{b_e^2(b_c\strqu + k_1)\cos^2 t}{b_c\strqu}
       - \Frac{k_0k_1(a_c\strqu + k_2)(b_c\strqu + k_2)}{a_c\strqu b_c\strqu k_2}
     \\[3.5mm]
    &= &-k_2 + k_0 + k_1 - \Frac{k_0 k_1}{k_2} = \Frac{(k_2-k_0)(k_1-k_2)}{k_2}\,,
  \end{array}
\]  
hence, in accordance with the third equation in \eqref{eq:gradients},
\begin{equation}\label{eq:|te|Raum}
  \Vert \Vkt t_e(t)\Vert = \frac{\sqrt{(k_0-k_2(t))(k_1-k_2(t))}}{\sqrt{-k_2(t)}}\,.
\end{equation}   
By virtue of \eqref{eq:k_e minus k_h}, this is also valid for $k_1 = 0\,$.
  
\begin{lem}\label{lem:Joachim_Raum} 
Let $e$ be the curve of intersection between the ellipsoid $\mathcal E$ with $k = k_0$ and semiaxes $(a_e,b_e,c_e)$ and the one-sheeted hyperboloid $\mathcal H_1$ with $k = k_1\,$. 
Then the  Joachimsthal integral for focal billiards along $e$ equals
\[  J_e = \frac{k_0 - k_1}{a_e b_e c_e}\,.
\]
\end{lem}

\begin{proof}
The definition of the Joachimsthal integral in \cite[p.~3]{Ako-Tab} yields for $e = \mathcal E\cap\mathcal H_1$  
\[  J_e:= -\langle \Vkt u_i,\,\Vkt n_{0|i} \rangle,
\]
where $\Vkt u_i$ as the unit vector of the directed side $P_iP_{i+1}$ and $\Vkt n_{0|i}$ the normalvector of $\mathcal E$ at $P_i$ (cf.\ \cite[Figure~3]{Sta_I}).
Thus, we obtain for focal billiards by virtue of \eqref{eq:wkl_raum}
\[  J_e = \Vert \Vkt n_0\Vert\,\sin\frac{\theta}2
    = \frac{\sqrt{(k_0-k_1)(k_0-k_2)}}{a_e b_e \sqrt{k_0}}\,
    \sqrt{\frac{k_0-k_1}{k_0-k_2}}\,,
\]
which confirms the claim.     
\end{proof}

The expression in \Lemref{lem:Joachim_Raum} yields for the planar billiard with an elliptic caustic, i.e., for $k_1 = 0$, the statement of \cite[Lemma~3.4]{Sta_I}. 
A similar result follows for the other planar limit $k_1 = -b_c^2$, namely for the hyperbolic billiard in $e''$ as  
\[  J_{e''} = \frac{b_e}{a_e \sqrt{k_e}} = \frac{\sqrt{k_e''}}{a_e'' b_e''}\,.
\]

\subsection{Periodic focal billiards in ellipsoids}

The focal billiards in the triaxial ellipsoid $\mathcal E$ (semiaxes $a_e,b_e,c_e$ with $a_e > b_e > c_e$) are periodic if and only if the isometric billiards in the ellipse $e'$ (semiaxes $a_e,b_e$) with the focal ellipse $c'$ (semiaxes $a_c = \sqrt{a_e^2-c_e^2}$ and $b_c = \sqrt{b_e^2-c_e^2}\,)$ as caustic are periodic.
We just recall from \cite[Theorem~2]{Sta_II} that, in terms of Jacobian elliptic functions to the modulus $m = d/a_c$ equal to the numerical eccentricity of $c'$, the periodicity is equivalent to a rational quotient
\begin{equation}\label{eq:rational}
  \frac K{\varDelta\mskip 1mu\wt u} \zwi\mbox{with}\zwi 
  \Cn\varDelta\mskip 1mu \wt u = \frac{b_c}{b_e}\zwi \mbox{and}\zwi 
  K = \int_0^{\pi/2} \frac{\mathrm d\varphi}{\sqrt{1 - m^2\sin^2\varphi}}\,. 
\end{equation}

Given any ellipsoid $\mathcal E$, there exists a two-dimensional set of inscribed focal billiards.
Each planar billiard in $e'$ determines by variation of $k_1$ a one-parameter set of isometric focal billiards inscribed respectively in lines of curvature $e\,$.
While the perimeter $L_{\mathcal E}$ and even the side lengths remain constant, the Joachimsthal integral $J_e$  depends on $k_1\,$.

On the other hand, the {\em billiard motion} along the planar billiard in $e'$ induces billiard motions along each single $e\,$, and this time $J_e$ remains constant, while angles and side lengths vary.
For the particular case of periodic billiards, it is natural to check which of the around 80 invariants as listed in \cite{80} have spatial counterparts at focal billiards.
We pick out a few of them.

We continue with the spatial analogue of a theorem proved in \cite[Theorem~4]{Ako-Tab}, that was a consequence of experiments carried out by D.\ Reznik \cite{80}.

\begin{thm}\label{lem:sum_cos_spatial} 
Let the focal billiards of the triaxial ellipsoid $\mathcal E$ with semiaxes $a_e,b_e,c_e$ be $N$-periodic.
Then for the billiards inscribed in the line of curvature $e = \mathcal E\cap\mathcal H_1$ with the constant elliptic coordinates $k_0$ and $k_1$, the sum of cosines of the exterior angles $\theta_i$ equals 
\[   \sum_{i=1}^N \cos\theta_i = N - J_e\,L_{\mathcal E} 
     = N - \frac{k_0-k_1}{a_eb_ec_e}\,L_{\mathcal E}
\]
with $J_e$ as Joachimsthal integral of $e$ and $L_{\mathcal E}$ as common perimeter of the focal billiards of $\mathcal E\,$.
\end{thm}

\begin{proof} 
In this proof we follow exactly the lines of the proof for the planar version in \cite{Ako-Tab}:
We compute the perimeter of the focal billiards in $\mathcal E$ as
\[ \begin{array}{rcl}
   L_{\mathcal E} &= &\Sum_{i=1}^N \langle(\Vkt p_{i+1} - \Vkt p_i),\,\Vkt u_i\rangle
      = \Sum_{i=1}^N \langle \Vkt p_{i+1},\,\Vkt u_i\rangle
        - \Sum_{i=1}^N \langle \Vkt p_i,\,\Vkt u_i\rangle
   \\[2.5mm]
   &= &\Sum_{i=1}^N \langle\Vkt p_i,\,\Vkt u_{i-1}\rangle 
      - \Sum_{i=1}^N \langle\Vkt p_i,\,\Vkt u_i\rangle    
       = \Sum_{i=1}^N \langle\Vkt p_i,\,(\Vkt u_{i-1} - \Vkt u_i)\rangle.
   \end{array}
\]
The signed angle between the unit vectors $\Vkt u_{i-1}$ and $\Vkt u_i$ equals $\theta_i$ (note \Figref{fig:symm_Erz}), and the difference vector $(\Vkt u_{i-1} - \Vkt u_i)$ has the direction of the normal vector $\Vkt n_{0|i}$ at $\Vkt p_i$ to $\mathcal E$, hence
\[  \Vkt u_{i-1} - \Vkt u_i = 2\sin\frac{\theta_i}2\,\frac 1{\Vert\Vkt n_{0|i}\Vert}
     \,\Vkt n_{0|i}\,, \ \mbox{where} \ \langle\Vkt p_i,\,\Vkt n_{0|i}\rangle = 1
\]
due to the equation of $\mathcal E$.
On the other hand, the Joachimsthal integral of $e=\mathcal E\cap\mathcal H_1$ equals  
\[  J_e = -\langle\Vkt u_i,\,\Vkt n_{0|i}\rangle = \Vert\Vkt n_{0|i}\Vert \sin\frac{\theta_i}2
\]
for all $i$, which results in
\[  J_e L_{\mathcal E} = \sum_{i=1}^N  2\sin^2 \frac{\theta_i}2
    = \sum_{i=1}^N (1 - \cos\theta_i).
\]
This completes the proof.  
\end{proof}

Now we transfer two results on elliptic billiards (note \cite[Theorems~8 and 9]{Sta_II}) to the isometric focal billiards .
The second result is a refinement of \Corref{cor:li/ri}.
       
\begin{thm}\label{thm:prod_li}
Let $P_1P_2\dots P_N$ be an $N$-periodic focal billiard in the triaxial ellipsoid $\mathcal E$ with the elliptic coordinate $k_0 = c_e^2\,$.
If the sides $[P_1,P_2], \,[P_2,P_3],\dots$ intersect the plane $z=0$ through the focal ellipse $c'$ in the points $Q_1, Q_2, \dots$, 
then the distances $r_i:= \ol{Q_{i-1}P_i}$ and $l_i:= \ol{P_iQ_i}$, $i = 1,\dots,N$, satisfy  
\[ \left. \begin{array}{rl}
    \mbox{for}\zwi N = 4n, \ n\in\NN:\zwi        &r_i\cdot r_{i+n} = l_i \cdot l_{i+n} 
    \\[1.0mm]
    \mbox{for}\zwi N = 4n\!+\!2, \ n\in\NN:\zwi  &r_i\cdot l_{i+n} = l_i \cdot r_{i+n+1}
  \end{array} \right\} = k_e
\]
and
\[  \prod_{i=1}^N \,l_i = \prod_{i=1}^N \,r_i = k_0^{N/2}, \zwi\mbox{and for}\zwi    N\equiv 0\mskip -10mu\pmod 4 \zwi\mbox{even}\zwi
    \prod_{i=1}^{N/2} \,l_i = \prod_{i=1}^{N/2} \,r_i = k_0^{N/4}.
\]
These formulas are also valid for the hyperbolic billiard $P_1''P_2''\dots P_N''$ inscribed in $e''$.
\end{thm}

\begin{rem} 
Similar results can be expected for the intersection points $T _1,T_2,\dots, T_N$ of the billiards' sides with the plane $y=0$ through the focal hyperbola $c''$.
\end{rem}

\subsection{Focal billiards and elliptic functions}

As shown in \cite[Theorem~1]{Sta_II} for elliptic billiards in $e'$, there is a one-parameter Liegroup of transformations which act not only on $e'$ but on all points of the associated Poncelet grid. 
It preserves confocal ellipses and permutes confocal hyperbolas and tangents of the caustic $c'$.
Hence, it induces billiard motions of all billiards with the common caustic $c'$.
A generating infinitesimal transformation is defined by a field of velocity vectors for all points in the exterior of $c'$ as
\[  P = (a_e\cos t,\,b_e\sin t) \,\mapsto\, \Vkt v_P = 
     \Vert\Vkt t_c(t)\Vert\,\Vkt t_e(t) =
     \sqrt{-k_h(t)}\,(-a_e\sin t,\,b_e\cos t).
\]     
If we see the velocity vectors as derivations of the position vectors by a parameter $u$ and use a dot for indicating this differentiation, then we obtain for the standard parameter $t$ and for the elliptic coordinates $(k_e,k_h)$ of $P$ in the plane $z=0$
\begin{equation}\label{eq:dot}
  \dot t = \sqrt{-k_h(t)}, \quad \dot k_e = 0,\quad
    \dot k_h = -2\,\sqrt{k_h(a_c\strqu + k_h)(b_c\strqu + k_h)}\,.
\end{equation}

By virtue of \Lemref{lem:corr_points}, we extend the field of velocity vectors to the space by differentiating $\Vkt e(t)$ in \eqref{eq:Bahn_Raumbil} and taking \eqref{eq:dot} into account.
In terms of spatial elliptic coordinates we obtain the assignment
\[  (k_0,k_1,k_2) \,\mapsto\, (\dot k_0,\,\dot k_1,\,\dot k_2)
    = \left(0,\,0,\,-2\sqrt{k_2(a_c\strqu + k_2)(b_c\strqu + k_2)}\,\right).
\]
The corresponding transformations preserve ellipsoids and one-sheeted hyperboloids in the confocal family, while they permute two-sheeted hyperboloids and the generators in each regulus of one-sheeted hyperboloids.
Thus, it induces billiard motions on all one-sheeted hyperboloids.

The parameter $u$ of the Liegroup in the plane as well as in space is {\em canonical} in the sense of \cite{Izmestiev}.
This means that on $e'$ the billiard transformation from one vertex to the next one corresponds to a shift of the respective canonical coordinates $u_i \mapsto u_{i+1} = u_i + 2\mskip 2mu\varDelta\mskip 1mu u\,$.
After the transference, the same is true for all focal billiards of the ellipsoid and for the hyperbolic billiard inscribed in $e''$.

\medskip
According to \cite[Theorem~2]{Sta_II}, a canonical parametrization can be expressed in terms of Jacobian elliptic functions to the modulus $d/a_c\,$.
After replacing $u$ by $\wt u:= a_c' u\,$, the three Jacobian elliptic base functions (see, e.g., \cite{Hoppe}) are  
\[  \SN\wt u = -\cos t, \quad \CN\wt u = \sin t,\quad
    \DN\wt u = \sqrt{1 - m^2\SN\wt u}\, 
     = \frac 1{a_c'}\,\sqrt{a_c\strqu\sin^2 t + b_c\strqu\cos^2 t}\,. 
\] 
The canonical parametrization of $e'$, namely $(-a_e\SN\wt u,\ b_e\CN\wt u)$, and the formula $k_2(\wt u) = k_h(\wt u) = - a_c\strqu \DN^2\wt u$ \,in \cite[(4.11)]{Sta_II} yield, due to \eqref{eq:Bahn_Raumbil}, the canonical parametrization
\begin{equation}\label{eq:e(u)}
  \Vkt e(\wt u) = \left( x(\wt u),y(\wt u),z(\wt u)\right) = \left(
   -\Frac{a_e \sqrt{a_c\strqu + k_1}}{a_c'}\SN\wt u, \ 
    \Frac{b_e \sqrt{b_c\strqu + k_1}}{b_c'}\CN\wt u, \
   \pm\Frac{\sqrt{-k_0 k_1}}{b_c'}\DN\wt u \right) 
\end{equation}
of the trajectory $e$ with the constant elliptic coordinates $k_0$ and $k_1\,$.
For the hyperbolic billiard as limit with $k_1 = -b_c\strqu$ follows by \eqref{eq:param hyp_bil}
\begin{equation}\label{eq:e"(u)}
  x''(\wt u) = \sqrt{\smFrac{a_c\strqu - b_c\strqu}{a_c\strqu + k_1}}\,x(\wt u)
          = -\frac{a_e d'}{a_c'}\,\SN\wt u\,, \zwi y'' = 0\,,\zwi
    z''(\wt u) = \frac{b_c'}{\sqrt{-k_1}}\,z(\wt u)  
         = \pm \sqrt{k_0}\,\DN\wt u\,.
\end{equation}

Each confocal ellipse $e'$ in the exterior of $c'$ is the trajectory of a billiard with the caustic $c'$.
By virtue of \cite[Theorem~2]{Sta_II}, the corresponding shift $\varDelta\wt u$ satisfies
\[ a_e = \frac{a_c'\Dn(\varDelta\wt u)}{\Cn(\varDelta\wt u)}\quad\mbox{and}\quad 
    b_e = \frac{b_c'}{\Cn(\varDelta\wt u)}\,.
\]
The conditions for a periodic billiard in $e'$ are given in \eqref{eq:rational}.

As explained in \cite[Corollary~3]{Sta_II}, it makes sense to replace $k_e = k_0$ as a coordinate for the confocal ellipses $e'$ by $\wt v:= \varDelta\wt u\,$, since it is canonical, too.
Note that
\[  k_0(\wt v) = a_e^2 - a_c\strqu = \frac{a_c^2\,\mathrm{sn}^2\wt v}
    {\mathrm{cn}^2\wt v}(1 - m^2).
\]     
By use of \Lemref{lem:corr_points}, we transfer the parametrization of the planar Poncelet grid by $(\wt u,\,\wt v)$ to the Poncelet grid on the one-sheeted hyperboloid $\mathcal H_1$ with the elliptic coordinate $k_1$ (Figures~\ref{fig:bil_mit_Hyp} and \ref{fig:bil_mit_Hyp2}).
Thus, we obtain

\begin{figure}[htb]
  \centering
  \def\sz{\small} 
  \psfrag{P1}[lb]{\contourlength{1.2pt}\contour{white}{\sz\red $P_1$}}
  \psfrag{P2}[rt]{\contourlength{1.2pt}\contour{white}{\sz\red $P_2$}}
  \psfrag{P3}[lb]{\contourlength{1.2pt}\contour{white}{\sz\red $P_3$}}
  \psfrag{P4}[rt]{\contourlength{1.2pt}\contour{white}{\sz\red $P_4$}}
  \psfrag{P8}[rt]{\contourlength{1.2pt}\contour{white}{\sz\red $P_8$}}
  \psfrag{P9}[lb]{\contourlength{1.2pt}\contour{white}{\sz\red $P_9$}}
  \psfrag{Q7,14}[lc]{\contourlength{1.2pt}\contour{white}{\sz\blue $Q_{7,14}$}}
  \psfrag{Q1,8}[lc]{\contourlength{1.2pt}\contour{white}{\sz\blue $Q_{1,8}$}}
  \psfrag{Q2,9}[lc]{\contourlength{1.2pt}\contour{white}{\sz\blue $Q_{2,9}$}}
  \psfrag{S8^2}[lc]{\contourlength{1.2pt}\contour{white}{\sz\green $S_8^{(2)}$}}
  \psfrag{S2^2}[lc]{\contourlength{1.2pt}\contour{white}{\sz\green $S_2^{(2)}$}}
  \psfrag{S10^2}[lc]{\contourlength{1.2pt}\contour{white}{\sz\green $S_{10}^{(2)}$}}
  \psfrag{S1^2}[lc]{\contourlength{1.2pt}\contour{white}{\sz\green $S_1^{(2)}$}}
  \psfrag{S9^2}[lc]{\contourlength{1.2pt}\contour{white}{\sz\green $S_9^{(2)}$}}
  \psfrag{S8^4}[rt]{\contourlength{1.2pt}\contour{white}{\sz\green $S_8^{(4)}$}}
  \psfrag{S2^4}[rt]{\contourlength{1.2pt}\contour{white}{\sz\green $S_2^{(4)}$}}
  \psfrag{S1^4}[rc]{\contourlength{1.2pt}\contour{white}{\sz\green $S_1^{(4)}$}}
  \psfrag{S9^4}[rc]{\contourlength{1.2pt}\contour{white}{\sz\green $S_9^{(4)}$}}
  \psfrag{T1}[rc]{\sz\blue $T_1$}
  \psfrag{T6}[rc]{\sz\blue $T_6$}
  \psfrag{T7}[rc]{\sz\blue $T_7$}
  \psfrag{T11}[lc]{\contourlength{1.2pt}\contour{white}{\sz\blue $T_{11}$}}
  \psfrag{T10}[lc]{\contourlength{1.2pt}\contour{white}{\sz\blue $T_{10}$}}
  \psfrag{T3}[lc]{\contourlength{1.2pt}\contour{white}{\sz\blue $T_3$}}
  \psfrag{u}[rt]{\sz $\wt u$}
  \psfrag{v}[rt]{\sz $\wt v$}
  \psfrag{-K}[rb]{\sz $-K$}
  \psfrag{3K}[lb]{\sz $3K$}
  \psfrag{c}[rc]{\sz\blue $\boldsymbol{c}$}
  \psfrag{e}[lc]{\sz\red $\boldsymbol{e}$}
  \psfrag{e2}[lc]{\sz\red $\boldsymbol{e}^{(2)}$}
  \psfrag{e4}[lc]{\sz\green $\boldsymbol{e}^{(4)}$}
  \includegraphics[width=120mm]{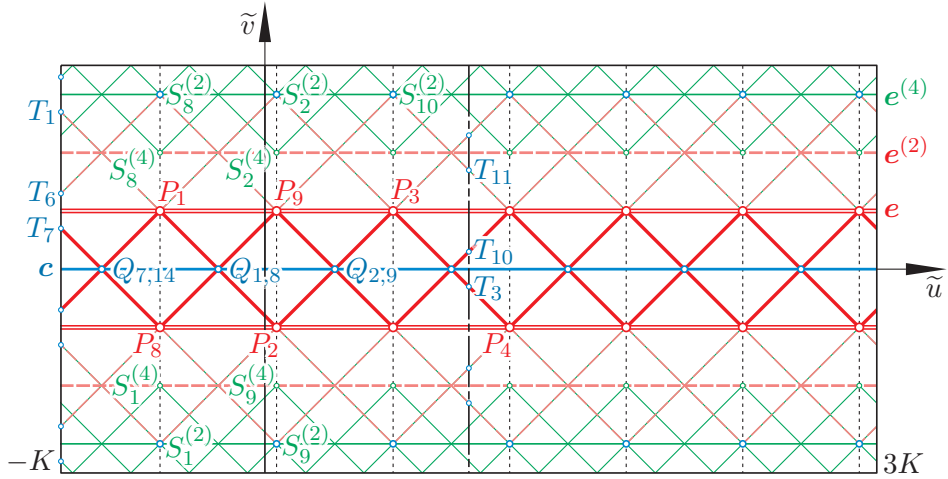} 
  \caption{The injective mapping $\Vkt Z$ sends the square grid of points $Q_i, P_i$ and $S_i^{j}$, $i=1,2\dots,14\,$, $j=1,2,4\,$, to the vertices and the diagonals to the net of curvature lines of the Poncelet grid associated to a $N$-periodic focal billiard on a one-sheeted hyperboloid with $N=14$ and turning number $\tau=2\,$.}
  \label{fig:map}
\end{figure}    

\begin{thm} 
Referring to the notation in \Thmref{thm:isometric2}, the injective mapping
\def\arraycolsep{1mm}
\[ \begin{array}{rl}
    \Vkt Z\!: &U\times V \to \mathcal H_1, \quad 
    (\wt u,\,\wt v) \,\mapsto\, \left(-a_{h_1}\Frac{\SN\wt u\,\DN\wt v}{\CN\wt v}
     \,, \ 
     b_{h_1}\Frac{\CN\wt u}{\CN\wt v}\,, \ 
     \pm\,c_{h_1}\Frac{\DN\wt u\,\SN\wt v}{\CN\wt v}\right)
    \\[3.0mm]
    &\mathrm{for}\zwi U:= \{\wt u \ |\,-K\le \wt u < 3K\},\zwi
       V:= \{\wt v \ |\,-K\le \wt v < K\} 
   \end{array} 
\]
parametrizes the one-sheetet hyperboloid $\mathcal H_1$ with semiaxes $(a_{h_1}, b_{h_1}, c_{h_1})$ 
in such a way, that the lines $\wt u = \const$ are branches of the lines of curvature with $k_2 = \const\,$.
Along $\wt v = \const$ we have $k_0 = \const$, and $\wt u \pm \wt v = \const$ defines generators of $\mathcal H_1\,$.\end{thm}

The mapping $\Vkt Z$ sends a square grid, for example that displayed in \Figref{fig:map}, to the Poncelet grid of a periodic focal billiard on $\mathcal H_1\,$.
The domain of the mapping $\Vkt Z$ can be extended to $\RR^2$ and satisfies
\[ \Vkt Z((\wt u + 4K),\,\wt v) = \Vkt Z(\wt u,\,(\wt v + 2K))
   = \Vkt Z(\wt u,\,\wt v). 
\]

At the limit $k_1 = 0$ we obtain the elliptic billiard in the plane of the focal ellipse $c'$; in this case $\Vkt Z$ is two-to-one, since $\Vkt Z(\wt u,\,-\wt v) = \Vkt Z(\wt u,\,\wt v)$.
At the other limit $k_1 = -b_c'$ holds $\Vkt Z(K +\wt u,\,\wt v) = \Vkt Z(K - \wt u,\,\wt v)$.
This corresponds to the fact, that the flexible Henrici hyperboloid is two-fold covered in its flat limiting poses (see \cite[Figure~2.51]{Quadrics}).

\section{Conclusion}
We reported about a remarkable relation between billiards in ellipses and Henrici's flexible hyperboloid:
There is a continuous transition between elliptic and hyperbolic billiards together with the respectively associated Poncelet grids via focal billiards of an triaxial ellipsoid.
This transition preserves not only the side lengths of the billiards but also the distances to the vertices of the Poncelet grids.
In this way, invariants of an elliptic billiard under its billiard motion can be transferred to invariants of focal billiards and a hyperbolic billiard.
A certain symmetry between the intersection points of the spatial billiard with the planes of the corresponding focal conics and related invariants will be subject of a future publication.

\section*{Acknowledgment}
This paper was inspired by Dan Reznik's computer experiments and his flair for geometric specialities.
The author is very grateful to him and to Ronaldo Garcia for interesting and fruitful discussions.


\end{JGGarticle}

\begin{thebibliography}{99} 
%

\bibitem{Ako-Tab}
{A.\ Akopyan, R.\ Schwartz, S.\ Tabachnikov:}
{\em Billiards in ellipses revisited.}
Eur.\ J.\ Math.\ 2020. \url{doi:10.1007/s40879-020-00426-9}

\bibitem{Barrow}
{J.\ Barrow-Green:} {\em ``Knowledge gained by experience'': Olaus Henrici -- engineer, geometer and maker of mathematical models.}
Hist.\ Math.\ {\bf 54}, 41--76 (2021). 

\bibitem{Bialy-Tab}
{M.\ Bialy, S.\ Tabachnikov:}
{\em Dan Reznik's identities and more.}
Eur.\ J.\ Math.\ 2020, \url{doi:10.1007/s40879-020-00428-7}






\bibitem{Chavez}
 A.\ Chavez-Caliz: {\em More about areas and centers of Poncelet polygons.}  
Arnold Math.\ J.\ {\bf 7}, 91--106 (2021).

\bibitem{DR_2006}
V.\ Dragovi\'c, M.\ Radnovi\'c: {\em Geometry of Integrable Billiards and Pencils of Quadrics.}
J.\ Math.\ Pures Appl.\ {\bf 85}/6, 758--790 (2006).

\bibitem{DR_russ}
V.\ Dragovi\'c, M.\ Radnovi\'c: {\em Integrable billiards and quadrics.}
Russian Math.\ Surveys {\bf 65}/2, 319--379 (2010).



\bibitem{Conics}
{G.\ Glaeser, H.\ Stachel, B.\ Odehnal:} {\em The Universe of Conics.} 
 From the ancient Greeks to 21\textsuperscript{st} century developments.
 Springer Spectrum, Berlin, Heidelberg 2016.



\bibitem{Hoppe}  
{R.\ Hoppe:} {\em Elliptische Integrale und Funktionen nach Jacobi.}
\url{http://www.dfcgen.de/wpapers/elliptic.pdf} (2004--2015), accessed May 2021.

\bibitem{Izmestiev}
{I.\ Izmestiev, S.\ Tabachnikov:} {\em Ivory's Theorem revisited.}
 J.\ Integrable Syst.\ {\bf 2}/1, xyx006 (2017), 
  \url{https://doi.org/10.1093/integr/xyx006}

\bibitem{Quadrics}
{B.\ Odehnal, H.\ Stachel, G.\ Glaeser:} {\em The Universe of Quadrics.} 
 Springer-Verlag GmbH Germany, Berlin, Heidelberg 2020.


\bibitem{80}
{D.\ Reznik, R.\ Garcia, J.\ Koiller:}
{\em Eighty New Invariants of N-Periodics in the Elliptic Billiard.}
\url{arXiv:2004.12497v11 [math.DS]}, 2020
 	

\bibitem{Sta_motion}
{H.\ Stachel:} {\em Moving ellipses on quadrics.} 
 {G, Slov.\ \v{C}as.\ Geom.\ Graf.\ {\bf 17}, no.~33, 29--42 (2020).}

\bibitem{Sta_I}
{H.\ Stachel:} {\em The Geometry of Billiards in Ellipses and their Poncelet Grids.} \url{arXiv:2105.03362 [math.MG]} (2021). 

\bibitem{Sta_II}
{H.\ Stachel:} {\em On the Motion of Billiards in Ellipses.}  
\url{arXiv:2105.03624v1 [math.DG]} (2021). 

\bibitem{Tabach}
{S.\ Tabachnikov:} {\em Geometry and Billiards.}
  American Mathematical Society, Providence/Rhode Island 2005. 

\bibitem{Henrici}
 {\em H.\ Wieners und P.\ Teutleins Sammlungen mathematischer Modelle.}
  2.~Ausg., B.G.\ Teubner, Leipzig, Berlin 1912.

\bibitem{Wunderlich}
{W.\ Wunderlich:} {\em Orthogonale Erzeugendenpolygone auf einschaligen Hyperboloiden.}
Monatsh.\ Math.\ {\bf 89}, 163--170 (1980).

\end{thebibliography}
\end{document}